\documentclass[a4paper,UKenglish,cleveref, autoref, thm-restate]{lipics-v2021}

\usepackage{tikz}
\usepackage{amsmath}
\usepackage{xspace}
\usetikzlibrary{shapes.geometric}
\usepackage{forest}
\usepackage{amssymb}
\usepackage{mathrsfs}
\let\oldtextsc\textsc
\renewcommand{\textsc}[1]{\textnormal{\oldtextsc{#1}}}
\DeclareFontFamily{U}{rsfs}{}
\DeclareFontShape{U}{rsfs}{m}{n}{<-> s*[0.9] rsfs10}{}

\title{Deciding the Common Fragment of CTL with Past and LTL}

\author{Massimo Benerecetti}{Università degli Studi di Napoli ``Federico II'', Italy}{massimo.benerecetti@unina.it}{https://orcid.org/0000-0003-4664-6061}{}
\author{Dario {Della Monica}}{Università degli Studi di Udine, Italy}{dario.dellamonica@uniud.it}{https://orcid.org/0000-0001-9743-665X}{}
\author{Angelo Matteo}{Università degli Studi di Udine, Italy}{angelo.matteo@uniud.it}{https://orcid.org/0009-0009-4663-1522}{}
\author{Fabio Mogavero}{Università degli Studi di Napoli ``Federico II'', Italy}{fabio.mogavero@unina.it}{https://orcid.org/0000-0002-5140-5783}{}
\author{Gabriele Puppis}{Università degli Studi di Udine, Italy}{gabriele.puppis@uniud.it}{https://orcid.org/0000-0001-9831-3264}{}

\authorrunning{M. Benerecetti, D. Della Monica, A. Matteo, F. Mogavero and G. Puppis} 

\Copyright{Massimo Benerecetti, Dario Della Monica, Angelo Matteo, Fabio Mogavero and Gabriele Puppis} 



\nolinenumbers 

\NewDocumentCommand\Dom{r()}{\mathop{Dom}(#1)}
\NewDocumentCommand\subtree{e{_} m}{#2_{#1}}  

\NewDocumentCommand\AP{}{\mathit{AP}}  

\NewDocumentCommand\Lang{r()}{\mathscr{L}(#1)}  

\let\oldtextsc\textsc
\renewcommand\textsc[1]{\text{\upshape\scshape #1}}

\newcommand{\MSO}{\textsc{MSO}\xspace}

\newcommand{\MTL}{\textsc{MTL}\xspace}
\newcommand{\MPL}{\textsc{MPL}\xspace}
\newcommand{\WMTL}{\textsc{WMTL}\xspace}

\newcommand{\WMPL}{\textsc{WMPL}\xspace}
\newcommand{\FO}{\textsc{FO}\xspace}
\newcommand{\WMSO}{\textsc{WMSO}\xspace}

\newcommand{\WCL}{\textsc{WMCL}\xspace}

\newcommand{\APCTL}{\textsc{APCTL}\xspace}
\newcommand{\EPCTL}{\textsc{EPCTL}\xspace}
\newcommand{\PCTLs}{\textsc{PCTL$^*$}\xspace}
\newcommand{\CTLs}{\textsc{CTL$^*$}\xspace}

\newcommand{\obbligationCTLs}{\textsc{ObligationCTL$^*$}\xspace}
\newcommand{\ACTL}{\textsc{ACTL}\xspace}

\newcommand{\PCTL}{\textsc{PCTL}\xspace}
\newcommand{\CTL}{\textsc{CTL}\xspace}
\newcommand{\LTL}{\textsc{LTL}\xspace}
\newcommand{\PLTL}{\textsc{PLTL}\xspace}

\newcommand{\safeLTL}{\textsc{SafeLTL}\xspace}
\newcommand{\cosafeLTL}{\textsc{coSafeLTL}\xspace}
\newcommand{\safePLTL}{\textsc{SafePLTL}\xspace}
\newcommand{\cosafePLTL}{\textsc{coSafePLTL}\xspace}

\newcommand{\Lmu}{\textsc{L$_\mu$}\xspace}
\newcommand{\AFLmu}{\textsc{AFL$_\mu$}\xspace}

\let\diamond\Diamond
\let\square\Box

\newcommand{\APT}{\textsc{APT}\xspace}
\newcommand{\ABT}{\textsc{ABT}\xspace}
\newcommand{\NBT}{\textsc{NBT}\xspace}
\newcommand{\AWT}{\textsc{AWT}\xspace}
\newcommand{\HWT}{\textsc{HWT}\xspace}

\newcommand{\NBW}{\textsc{NBW}\xspace}
\newcommand{\DBW}{\textsc{DBW}\xspace}
\newcommand{\UWW}{\textsc{UWW}\xspace}
\newcommand{\NBWcf}{\textsc{NBW$_{cf}$}\xspace}
\newcommand{\UBWcf}{\textsc{UBW$_{cf}$}\xspace}
\newcommand{\NCWcf}{\textsc{NCW$_{cf}$}\xspace}
\newcommand{\UCWcf}{\textsc{UCW$_{cf}$}\xspace}

\newcommand{\UBW}{\textsc{UBW}\xspace}

\newcommand{\NCW}{\textsc{NCW}\xspace}
\newcommand{\UCW}{\textsc{UCW}\xspace}

\newcommand{\HWTcf}{\textsc{HWT$_{cf}$}\xspace}
\newcommand{\HTcf}{\textsc{HT$_{cf}$}\xspace}
\newcommand{\DBWcf}{\textsc{DBW$_{cf}$}\xspace}

\newcommand{\X}{\mathsf{X}\xspace}

\newcommand{\F}{\mathsf{F}\xspace}
\newcommand{\G}{\mathsf{G}\xspace}
\newcommand{\U}{\mathsf{U}\xspace}

\newcommand{\R}{\mathsf{R}\xspace}

\newcommand{\Y}{\mathsf{Y}\xspace}

\newcommand{\s}{\mathsf{S}\xspace}

\newcommand{\PP}{\mathsf{P}\xspace}

\newcommand{\E}{\mathsf{E}\xspace}
\newcommand{\A}{\mathsf{A}\xspace}
\newcommand{\D}{\mathsf{D}\xspace}

\newcommand{\membership}{\mapsto}

\begin{document}

\maketitle


\begin{abstract}
A central goal of language theory is to compare formalisms by understanding both their expressive overlaps and their relative expressive power.
One particularly challenging question in this direction is the problem of determining the \emph{common fragment} 
of two formalisms $F_1$ and $F_2$, that is, effectively characterise the class $F_1\cap F_2$ of properties that 
can be expressed in both formalisms. This question can be equally phrased as a decision problem: given 
a property expressed in $F_1$ or $F_2$, decide whether the same property can be also expressed in $F_1\cap F_2$. 
A question closely related to this is the \emph{membership problem}, denoted $F_1 \membership F_2$,
which asks whether a property expressed in $F_1$ can be also expressed in $F_2$.
These problems become particularly difficult when \emph{branching-time} formalisms are involved,
in general due to the lack of equivalent algebraic characterizations.

In this work, we prove that $\LTL \cap \PCTL$ is decidable, where \PCTL denotes
\CTL extended with \emph{past operators}.
We do this by showing that both membership problems, $\LTL \membership \PCTL$ and 
$\PCTL \membership \LTL$, are decidable.
The direction $\PCTL \membership \LTL$ follows from suitable combinations of known results. 
The converse direction, $\LTL \membership \PCTL$, requires an automata-theoretic characterisation
of $\PCTL$. 
Specifically, we introduce a new class of automata, called \emph{counter-free hesitant weak tree automata} ($\HWTcf$)
that capture precisely the expressiveness of $\PCTL$, and that are obtained by combining two orthogonal restrictions
on alternating parity tree automata, namely, \emph{counter-free hesitancy} and \emph{weakness}. 
We then prove that, for every word language $L$ defined by an \LTL formula,
the associated tree language $\triangle[L]$ is recognisable by an \HWTcf
if and only if $L$ is recognized by a deterministic B\"uchi word automaton.
Since the latter recognisability problem is known to be decidable, so is the
former.

This result advances the longstanding open problem of deciding $\LTL \cap \CTL$. 
Indeed, that problem can now be reduced to $\PCTL \membership \CTL$, that is, the 
question of when past operators can be eliminated.

\end{abstract}

\newpage


\section{Introduction}\label{sec:intro}

The notion of a \emph{common fragment} is central to expressiveness theory: 
given two formalisms $F_1$ and $F_2$, the aim is to effectively characterise 
the class $F_1 \cap F_2$ of properties that are definable in both.
Closely related is the associated \emph{membership problem} $F_1 \membership F_2$: 
given a specification in $F_1$, determine whether it can be transformed into an 
equivalent specification in $F_2$.
Of course, deciding the common fragment $F_1 \cap F_2$ reduces to solving
two membership problems, i.e.~$F_1 \membership F_2$ and $F_2 \membership F_1$.

These questions have a long history in language theory, tracing back at least to
Sch\"utzenberger's theorem on \emph{star-free word languages} and
\emph{aperiodic monoids}~\cite{Sch65} and to Kamp's theorem connecting Prior's
\emph{tense logic}~\cite{Pri67} with \emph{first-order logic} over
\emph{linearly ordered structures}~\cite{Kam68}.
From these seminal works, a rich line of research has emerged, yielding a
remarkably complete picture for languages of finite and infinite words, where
predicate logics~\cite{Buc60,Buc62,Buc66,Lau68}, temporal
logics~\cite{Pnu77,Pnu81,Wol83}, automata~\cite{McN66,MP71,Cho74,Lad77}, and
algebraic structures~\cite{PP86} are tied together by deep and elegant
correspondences~\cite{PP04}.

%

By contrast, for regular tree languages the situation is considerably more 
difficult and far more fragmented.
A milestone is Rabin's work on the monadic second-order theory of infinite trees, 
which established a decisive connection between \emph{Monadic Second-Order Logic} ($\MSO$) 
and tree automata~\cite{Rab70}. This connection made it possible to study logical questions 
about decidability and definability by means of automata.
Within this framework, full $\MSO$ corresponds to automata with parity conditions, 
whereas weaker logics correspond to weaker acceptance mechanisms. 
In particular, \emph{Weak $\MSO$} ($\WMSO$), where second-order quantification 
ranges only over finite sets, admits an automaton-based characterisation 
in terms of \emph{alternating weak tree automata} ($\AWT$)~\cite{MSS92},
a subclass of alternating parity tree automata where each strongly connected 
component of the transition graph has a single priority.
Thanks to this characterisation, the membership problem $\MSO \membership \WMSO$ 
can be viewed automata-theoretically as the problem of deciding when a parity 
tree automaton can be replaced by an equivalent weak one.
The same parity-versus-weakness boundary also appears with bisimulation-invariant
properies, where it governs the comparison between the \emph{modal $\mu$-calculus} ($\Lmu$)~\cite{Koz83} 
and its \emph{alternation-free fragment} ($\AFLmu$).
More precisely, thanks to the celebrated characterisation theorem of Janin and Walukiewicz~\cite{JW96},
$\Lmu$ captures precisely the bisimulation-invariant fragment of $\MSO$, and 
admits a standard presentation in terms of alternating parity tree automata~\cite{Wil01}.
Its alternation-free fragment, in turn, corresponds to weak alternating automata,
as well as to the bisimulation-invariant fragment of \WMSO.
Thus, the question whether a $\Lmu$ property is definable in $\AFLmu$ can again 
be reduced to a membership problem between classes of automata, namely, $\APT \membership \AWT$; 
in other words, deciding when parity can be replaced by weakness.
This viewpoint has become central in verification: many logical formalisms 
can be analysed via their corresponding automata models, and suitable restrictions 
can often be identified, both at the logic level and at the automaton level, 
that support a number of effective verification procedures~\cite{KV98}.

Within this landscape, the two formalisms most relevant to us are \AFLmu and \CTLs. 
Over infinite binary trees, these formalisms correspond to natural fragments of \MSO.
Indeed, we have already seen that \AFLmu corresponds to \WMSO (note that the 
bisimulation relation is trivial over binary trees). As for \CTLs, this corresponds 
to the \emph{path-quantification fragment} of \MSO (\MPL)~\cite{HT87,MR99,MR03}.
Moreover, recent automata-theoretic characterisations of \CTLs have made this picture 
even sharper, by identifying \emph{counter-free hesitancy} as the structural
constraint on alternating parity tree automata that achieves the same expressive power 
as \CTLs~\cite{BBMP24}.
It is therefore natural to ask for a formalism $L$ that is as expressive as
$\WMSO \cap \MPL$, or, equally, as $\AFLmu \cap \CTLs$.
At the same time, this question appears to be surprisingly open.
Identifying such a logic would have immediate consequences.
First, analogous to Rabin's results for $\MSO \membership \WMSO$~\cite{rabin1970weakly}, this approach would yield a reduction from the membership problem $\MPL \membership \WMSO \cap \MPL$ to the problem $\APT \membership \AWT$.
Second, and more importantly for us, it would provide a new perspective on the
closely related and longstanding open problem of determining the common fragment
$\LTL \cap \CTL$~\cite{EH86}, where $\LTL$ here is interpreted over trees using 
an implicit universal path-quantification.
In other words, the problem of characterising $\LTL \cap \CTL$ amounts to 
understanding the overlap between a linear-time and a branching-time formalism.

The history of the common fragment $\LTL \cap \CTL$ spans almost four decades.
The first major result is the decidability of $\CTLs \mapsto \LTL$, due to
Clarke and Draghicescu, which also settles $\CTL \mapsto \LTL$~\cite{CD88}.
A possible route towards $\CTLs \mapsto \CTL$, conjectured in the same work, was
later ruled out by Boker and Shaulian~\cite{BS18}.
The inverse membership problem, i.e., ~$\LTL \mapsto \CTL$, proved substantially 
more difficult.
Maidl gave a necessary and sufficient condition for $\LTL \mapsto
\ACTL$~\cite{Mai00}, where $\ACTL$ is the syntactic fragment of $\CTL$ in which the existential path quantifier $\E$ is disallowed, and Boja\'nczyk later proved this problem decidable, while
also showing the strict inclusion $\LTL \cap \ACTL \subsetneq \LTL \cap \CTL$.
This is conceptually striking: although $\LTL$ expresses only universal
properties, existential path quantification is still needed on the $\CTL$ side
to capture the full common fragment~\cite{Boj08}.
A further key development is the result of Kupferman and Vardi showing that
$\LTL \mapsto \AFLmu$ is decidable, with an \textsc{ExpSpace} upper bound and a
\textsc{PSpace} lower bound~\cite{KV05}.

The present paper develops this automata-theoretic perspective. 
We introduce a new class of automata, called \emph{counter-free hesitant weak tree automata} ($\HWTcf$). 
The definition combines two orthogonal restrictions: \emph{counter-free hesitancy}, 
which comes from the automata model for $\CTLs$, and \emph{weakness}, which underlies 
the automata model for $\AFLmu$.
This combination makes $\HWTcf$ a natural candidate for capturing the common fragment 
$\CTLs \cap \AFLmu$, and it leads to the following conjecture:

\begin{conjecture}
\label{con:comfrg}
  \HWTcf are as expressive as $\CTLs \cap \AFLmu$, and, equivalently, as $\MPL \cap \WMSO$.
\end{conjecture}

This conjecture provides the conceptual starting point of the paper. 
Even without resolving it in full, the automata class $\HWTcf$ already 
reveals substantial structural information about the target intersection 
and leads to a sequence of results that we outline below.

Our first main result is a logical characterisation of $\HWTcf$ in terms of $\PCTL$, 
namely, $\CTL$ extended with \emph{past operators}. This identifies $\PCTL$ as a 
natural temporal-logic candidate for the common fragment under investigation.
The characterisation is supported by a normal form for $\PCTL$ and by the 
identification of an equivalent fragment of $\CTLs$ in which path formulae 
are restricted to define safety and co-safety properties.
The correspondence also preserves the polarity of path quantification. 
More precisely, $\HWTcf$ that only use universal non-determinism 
(we call them \emph{universal \HWTcf} for short) translate into $\APCTL$, 
i.e.~the fragment of $\PCTL$ that uses only universal path quantification.
Symmetrically, \emph{existential $\HWTcf$} translate into $\EPCTL$, 
i.e.~the fragment of $\PCTL$ that uses only existential path quantification.

\begin{restatable}{theorem}{autopctl}
\label{thm:hwtcf&pctl}
  \HWTcf and \PCTL are effectively equivalent formalisms.
  Moreover, universal \HWTcf (resp., existential \HWTcf) 
  admits translations into \APCTL (resp., \EPCTL).
\end{restatable}


A first consequence of the above characterisation is conditional 
on the conjectured equivalence between $\HWTcf$ and $\CTLs \cap \AFLmu$. 
Indeed, under the assumption that $\HWTcf = \CTLs \cap \AFLmu$, 
the above characterization would make $\PCTL$ the natural logical 
presentation for the common fragment $\CTLs \cap \AFLmu$.
Moreover, the membership problem $\CTLs \membership \PCTL$ 
would admit a Rabin-style automata-theoretic reduction: 
indeed, the problem would become equivalent to asking whether a 
$\CTLs$-definable language is also definable in $\AFLmu$ ($=\CTLs \cap \AFLmu$),
which in turn is reduced to asking whether a given
parity tree automaton (capturing a specification in $\CTLs$) 
is equivalent to some weak alternating tree automaton (capturing
a specification in $\AFLmu$). This gives the following conditional result:

\begin{restatable}{theorem}{mainthm}
\label{thm:rabthm}
  Assuming Conjecture~\ref{con:comfrg}, 
  the membership problem $\CTLs \mapsto \PCTL$ (or, equally, $\MPL \mapsto \PCTL$)
  reduces to the automata-theoretic problem $\APT \mapsto \AWT$.
\end{restatable}

\begin{figure}[t]
\centering
\begin{tikzpicture}[scale=0.7, transform shape]
  \def\balloonA{(-2.25,  0) ellipse (9.6em and 3em)}
  \def\balloonB{( 1.875,  0) ellipse (8.625em and 3em)}
  \def\balloonC{(-0.975,  0) ellipse (4.125em and 1.75em)}

  \begin{scope}
    \clip \balloonA;
    \fill[yellow!30] \balloonB;
  \end{scope}

  \begin{scope}
    \clip \balloonB;
    \fill[gray!30] \balloonC;
  \end{scope}

  \draw[semithick, black!80] \balloonA;
  \draw[semithick, black!80] \balloonB;
  \draw[semithick, black!80] \balloonC;

  \node at (-4.125, 0) {$\CTLs = \MPL$};
  \node at ( 3.225, 0) {$\AFLmu = \WMSO$};
  \node at (-1.725, 0) {$\LTL$};

  \coordinate (yellowarea) at (0.75,   0);
  \coordinate (grayarea)   at (-0.375, 0);

  \node[align=center] (yellowtext) at (0.75, 1.75)
    {$\PCTL \subseteq \CTLs \cap \AFLmu \stackrel{?}{\subseteq} \PCTL$};

  \node[align=center] (graytext) at (-0.75, -1.75)
    {$\LTL \cap \PCTL = \LTL \cap \APCTL = \LTL \cap \AFLmu$};

  \draw[->, thin] (yellowarea) to[out= 45, in=-135] (yellowtext.south);
  \draw[->, thin] (grayarea)   to[out=-45, in= 135] (graytext.north);
\end{tikzpicture}
\caption{Main relationships between classes of tree languages; new results concern shaded regions.}
\label{fig:venndiag}
\end{figure}
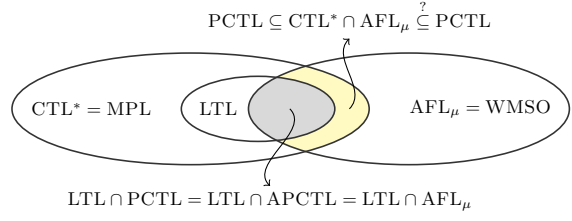

More importantly, the characterisation of \HWTcf in terms of \PCTL yields
unconditional consequences for the linear/branching-time interface.
Specifically, it enables an exact identification of the intersection of
\LTL with several branching-time formalisms: \PCTL, \APCTL, and \AFLmu.
Indeed, all these common fragments coincide.
This is particularly noteworthy because, unlike the situation for \CTL, the
presence of past operators makes existential path quantification unnecessary: 
in this respect, \PCTL behaves more like \CTLs than like \CTL.

\begin{restatable}{theorem}{comfrg}\label{thm:comfrg}
 $\LTL \cap \APCTL = \LTL \cap \PCTL = \LTL \cap \AFLmu$.
\end{restatable}

This equivalence can in turn be converted into decidability of
membership problems beetween \LTL and \PCTL.
More precisely, for an \LTL-definable word language $L$, we prove that $L$ is
recognisable by a deterministic Büchi automaton (\DBW) if and only if the 
associated tree language $\triangle[L]$ is recognisable by an \HWTcf automaton, 
equivalently definable in \PCTL.
Combining this correspondence with the known decidability of
\DBW-recognisability for \LTL~\cite{KV05} yields decidability of 
$\LTL \mapsto \PCTL$, with an \textsc{ExpSpace} upper bound and a \textsc{PSpace} lower bound.
The converse problem, namely $\PCTL \mapsto \LTL$, is also decidable, by means
of the equivalent \CTLs fragment identified above and the corresponding
characterisation of safety and co-safety path languages.

\begin{restatable}{theorem}{memprb}
\label{thm:memprb}
  The membership problem $\LTL \mapsto \PCTL$ is decidable in \textsc{ExpSpace}
  and is \textsc{PSpace}-hard.
  Moreover, $\PCTL \mapsto \LTL$ is decidable.
\end{restatable}

These results have a direct bearing on the longstanding open problem of deciding
$\LTL \mapsto \CTL$.
Indeed, once the common fragment of \LTL with \PCTL has been identified and its
membership problem shown decidable, the original question can be reformulated in
terms of \emph{past elimination}.
More precisely, deciding whether an \LTL formula is equivalent to some \CTL
formula reduces to deciding whether an equivalent \PCTL formula exists and 
can be translated into \CTL.
In this way, the classical problem of understanding $\LTL \cap \CTL$ is reduced
to the branching-time problem $\PCTL \mapsto \CTL$, thereby opening an
alternative route towards its solution through the study of elimination of past
operators.

\begin{restatable}{corollary}{ltlctl}
\label{cor:ltl2ctl}
  Decidability of $\PCTL \mapsto \CTL$ implies decidability of $\LTL \mapsto
  \CTL$.
\end{restatable}

Finally, we investigate whether \PCTL admits a matching characterisation in
terms of known monadic predicate logics.
Here the answer is negative.
On the one hand, \WMPL, which coincides with \FO, is known to be strictly less
expressive than \PCTL. 
On the other hand, the recently introduced tree fragment of \WMSO, namely \WMTL \cite{BBMP23},
turns out to be strictly more expressive, and in fact incomparable with \CTLs.
This shows that the currently known fragments of \WMSO do not capture the
expressive power of \PCTL, and therefore do not provide a purely monadic
characterisation of the target intersection. Conversely, looking beyond the expressive power of \WMSO is unnecessary; the inclusion $\PCTL \subseteq \WMSO \cap \MPL$ immediately implies that \PCTL is already a fragment of \WMSO.

\begin{restatable}{theorem}{wmtlmpl}
\label{thm:wmtl2mpl}
\WMTL is incomparable to both \CTLs and \MPL. Consequently, no known fragment of \WMSO is expressively equivalent to \PCTL.
\end{restatable}


\section{Preliminaries}

\subparagraph{Trees.}
We mainly work with trees, precisely with unordered, full binary trees. 
Formally, a tree is a pair $t=(\Dom(t),E)$, where $\Dom(t)$ is the set of nodes 
and $E \subseteq \Dom(t)\times \Dom(t)$ is the directed edge relation, such that:
\begin{enumerate}
\item there is a unique node $r$, called the root of $t$, from which all nodes can be reached, i.e., $(r,u)\in E^*$ for all $u\in\Dom(t)$;
\item for every $u\neq r$, there is a unique $v$, called the parent of $u$, such that $(v,u)\in E$;
\item the root has no parent, that is, there is no $v\in\Dom(t)$ with $(v,r)\in E$;
\item for every node $u$, there are exactly two nodes $v_1,v_2$ such that $(u,v_1),(u,v_2)\in E$;
      these are the children of $u$ (note that there is no particular order imposed on $v_1,v_2$).
\end{enumerate}
Unless otherwise stated, all trees are infinite and without leaves.
We use binary trees rather than unranked trees, where nodes may have an arbitrary finite number of children, 
in order to keep a direct correspondence between temporal and monadic logics. 
As a matter of fact, on unranked trees, temporal logics correspond only to bisimulation-invariant fragments 
of monadic logics, unless extended with counting operators (see, e.g., \cite{MR03}). 
Restricting the branching degree allows us to bypass these technicalities, though it is understood 
that our results readily extend to the unranked case. 

A \emph{$\Sigma$-tree} is a tree $t$ expanded with a labelling $\lambda: \Dom(t) \to \Sigma$, 
associating a letter from the finite alphabet $\Sigma$ with each node.
An \emph{$\omega$-tree language} is a set of 
$\Sigma$-trees. 
We denote by $\epsilon_t$ the root of a tree (or a $\Sigma$-tree) $t$. 
A \emph{path} of $t$ is a finite or infinite sequence $\pi = v_1 v_2 \dots$ 
of nodes that are two-by-two connected by the child relation, i.e., $(v_i,v_{i+1})\in E$ for all $i\ge1$.
A \emph{maximal path} is a path that cannot be extended further to the right, so it must be infinite
unless the tree has some leaves. 
Note that path and maximal paths may start at arbitrary nodes of trees.
A \emph{subtree} of $t$ is a finite or infinite set $D\subseteq\Dom(t)$ such that 
$E$ restricted to $D$ defines a tree, not necessarily full or leafless.

\subparagraph{Monadic logics.}
We introduce \emph{monadic second-order logic} (\MSO) on trees. 
Fix a finite set $\AP$ of atomic propositions, and let $\Sigma = 2^{\AP}$.
MSO formulae are defined by the grammar
\[
\varphi ~::=~ 
  x = x'
  ~\bigm|~
  x \le x'
  ~\bigm|~
  x \in X
  ~\bigm|~
  p(x)
  ~\bigm|~
  \neg \varphi
  ~\bigm|~
  \varphi \lor \varphi
  ~\bigm|~
  \exists x. ~ \varphi
  ~\bigm|~
  \exists X. ~ \varphi
\]
where $p\in\AP$ and 
$x,x'$ (resp.,~$X$) are first-order (resp., monadic) variables. 
First-order variables are interpreted as single nodes of a tree, while monadic
variables are interpreted as sets of nodes.
The atom $x \le x'$ states that the node assigned to $x'$ is a descendant of the node assigned to $x$;
the atom $x \in X$ states that the node assigned to $x$ belongs to the set assigned to $X$;
and $p(x)$ states that the node assigned to $x$ carries the proposition $p$, that is, 
$p$ belongs to its label.
The rest of the semantics is standard.

We consider the usual notion of \MSO \emph{sentence}. Since an \MSO sentence has no free variables, its truth value on a $\Sigma$-tree is independent of any variable assignment.
We write $t \models \varphi$ when $t$ is a $\Sigma$-tree that satisfies the sentence $\varphi$. 
An $\omega$-tree language $T$ is \emph{\MSO-definable} if there is a sentence $\varphi$ 
such that $T = \{ t \mid t \models \varphi \}$. 
Due to the correspondence between \MSO and several classes of automata on $\Sigma$-trees
and due to the closure properties satisfied by these classes \cite{Rab69,Mos84,MS87,Tho97}, 
\MSO-definable $\omega$-tree languages are also often called \emph{regular}.

Finally, we introduce the following fragments of \MSO:
\begin{itemize}
    \item \emph{Weak monadic second-order logic} (\WMSO) is the fragment in which
          monadic variables can only be assigned \emph{finite sets} of nodes of a tree. 
    \item \emph{Monadic tree logic} (\MTL), resp., \emph{weak \MTL} (\WMTL) are the fragments in which 
          monadic variables can only be assigned \emph{subtrees}, resp., \emph{finite subtrees}. 
    \item \emph{Monadic path logic} (\MPL), resp., \emph{Weak \MPL} (\WMPL) are the fragments in which 
          monadic variables can only be assigned \emph{paths}, resp., \emph{finite paths}.
    \item \emph{First order logic} (\FO) is the fragment in which monadic variables are disallowed.
\end{itemize}
\subparagraph{Temporal Logics.}

We consider here 
\emph{full computation tree logic with past}, denoted \PCTLs. 
Its formulae are generated in negation normal form according to the following grammar:
\begin{align*}
\varphi ~::= &~
  p
  ~\bigm|~
  \neg p
  ~\bigm|~
  \varphi \lor \varphi
  ~\bigm|~
  \varphi \land \varphi
  ~\bigm|~
  \E \psi
  ~\bigm|~
  \A \psi
\\
\psi ~::= &~
  \varphi
  ~\bigm|~
  \psi \lor \psi
  ~\bigm|~
  \psi \land \psi
  ~\bigm|~
  \X \psi 
  ~\bigm|~
  \Y \psi
  ~\bigm|~
  \tilde{\Y} \psi 
  ~\bigm|~
  \psi \U \psi
  ~\bigm|~
  \psi \R \psi
  ~\bigm|~
  \psi \s \psi
\end{align*}
Formulae generated from $\varphi$ are called \emph{state formulae}, 
while those generated from $\psi$ are called \emph{path formulae}. 
The operators $\E$ and $\A$ quantify existentially and universally 
over paths. 
The operators $\X$, $\U$, and $\R$ are future temporal operators: 
$\X$ is the usual ``next'' operator, and $\U$ and $\R$ are ``until'' 
and ``release''. Dually, $\Y$, $\tilde{\Y}$, and $\s$ are past 
temporal operators: $\Y$ is ``yesterday'', $\tilde{\Y}$ is 
``weak yesterday'', and $\s$ is ``since''. 
The difference between $\Y$ and $\tilde{\Y}$ is that $\Y$ 
requires the existence of the predecessor (i.e., parent) along the current path, 
whereas $\tilde{\Y}$ is also satisfied at the beginning of the path.
Unless explicitly stated otherwise, by a ``$\PCTLs$ formula'' 
we mean a \emph{state} $\PCTLs$ formula.

State formulae are interpreted over \emph{pointed trees},
i.e., pairs $(t,u)$ consisting of a $\Sigma$-tree and a node $u$ in it, where $\Sigma=2^{\AP}$. 
Similarly, path formulae are interpreted over \emph{path-pointed trees}, i.e., triples $(t,\pi,i)$
consisting of a $\Sigma$-tree, an infinite path $\pi$ in it, and a position 
$i\ge0$ on that path. 
We give the semantics only for the most interesting cases: 
\begin{itemize}
\item $(t,u) \models \E \psi$ iff there exists an infinite path $\pi$ starting at $u$ such that $(t,\pi,0) \models \psi$,
\item $(t,\pi,i) \models \psi_1 \U \psi_2$ iff $\exists j \ge i$ such that $(t,\pi,j) \models \psi_2$ and $\forall k. (i \le k < j), (t,\pi,k) \models \psi_1$,
\item $(t,\pi,i) \models \psi_1 \s \psi_2$ iff $\exists j \le i$ such that $(t,\pi,j) \models \psi_2$ and $\forall k. (i \ge k > j), (t,\pi,k) \models \psi_1$,
\end{itemize}
\noindent 
We write $t \models \varphi$ to mean $(t,\epsilon_t) \models \varphi$ --- in this case, we also say that
$t$ is a \emph{model} of $\varphi$.
We denote by $\Lang(\varphi)$ the set of all models of $\varphi$, and 
declare two formulae $\varphi,\varphi'$ \emph{equivalent} if $\Lang(\varphi) = \Lang(\varphi')$.
We will also consider the following fragments of \PCTLs: \CTLs, which is just \PCTLs without past operators; \PCTL, i.e., the variant of \PCTLs in which every future temporal operator (i.e., $\X$, $\U$ and $\R$) must be immediately preceded by a path quantifier; \CTL, i.e., \PCTL without past operators. By \cite[Theorem 3.6]{HT87}, we have that \MPL, \PCTLs and \CTLs are equivalent formalisms. 
We spell out the syntax of $\PCTL$ formulas in negation normal form, since this is the formalism that is mostly used throughout the paper:
\begin{align*}
\varphi ~::= &~
  p
  ~\bigm|~
  \neg p
  ~\bigm|~
  \varphi \lor \varphi
  ~\bigm|~
  \varphi \land \varphi
  ~\bigm|~
  \E\X \psi
  ~\bigm|~
  \E\psi\U\psi
  ~\bigm|~
  \E\psi\R\psi
  ~\bigm|~
  \A\X \psi
  ~\bigm|~
  \A \psi \U \psi
  ~\bigm|~
  \A \psi \R\psi
\\
\psi ~::= &~
  \varphi
  ~\bigm|~
  \psi \lor \psi
  ~\bigm|~
  \psi \land \psi
  ~\bigm|~
  \Y \psi
  ~\bigm|~
  \tilde{\Y} \psi
  ~\bigm|~
  \psi \s \psi.
\end{align*}

We also consider \emph{linear-time formalisms}. 
In particular, we consider the well-known \emph{linear temporal logic with past} (\PLTL), here defined as the set of path formulae of \PCTLs in which we forbid the use of path quantifiers, and some of its fragments:
\LTL is \PLTL without past operators; \cosafeLTL (resp., \safeLTL) is \PLTL in which past operators and the future operator $\R$ (resp., $\U)$ are disallowed; \cosafePLTL (resp., \safePLTL) consists of \PLTL formulae of the form $\F\alpha$ (resp., $\G\alpha$), where $\F\alpha = \top \U \alpha$ (resp., $\G\alpha = \bot \R \alpha)$ and $\alpha$ contains only past operators. 

\begin{theorem}[{\cite[Theorems 1 and 8]{chang1992characterization}}] \label{safeLTL equivalent to safePLTL}
\safePLTL is as expressive as \safeLTL; similarly, \cosafePLTL is as expressive as \cosafeLTL.
\end{theorem}

Theorem \ref{safeLTL equivalent to safePLTL} can be strenghtened to show that safety and co-safety properties can be
equally captured by slightly more relaxed formalisms than $\safePLTL$ and $\cosafePLTL$.

\begin{restatable}{proposition}{PsafeLTL} \label{safeLTL with past is equal to safeLTL}
The fragment of $\PLTL$ that consists of formulae in negation normal form and forbids the use of $\R$ (resp., $\U$) is
as expressive as $\cosafePLTL$ (resp., \safePLTL).
\end{restatable}

Finally, we interpret linear-time formalisms, in particular $\LTL$, over trees using the 
standard universal-path semantics. Thus, for an $\LTL$ formula $\varphi$, we write 
$t \models \varphi$ iff every path of $t$ starting at the root satisfies $\varphi$, 
or equivalently iff $t \models \A\varphi$.


\subparagraph{Word automata.} 
An \emph{$\omega$-word} over $\Sigma$ is an infinite sequence $\sigma_1 \sigma_2 \dots$ of letters from $\Sigma$.
An \emph{$\omega$-word language} is a set of $\omega$-words. We assume the reader to be familiar with
the theory of regular $\omega$-word languages \cite{Buc62}.

An $\omega$-word automaton is a tuple $\mathcal{A} = \langle Q, \Sigma, \delta, q_I, F \rangle$, where 
$Q$ is a finite set of states, 
$\Sigma$ is the alphabet, 
$q_I \in Q$ is an initial state, and 
$F \subseteq Q$ is a set of final states. 
The transition function is $\delta: Q \times\Sigma\to Q$ for a deterministic automaton, and 
$\delta: Q\times\Sigma\to 2^Q$ for a non-deterministic one.
Here non-determinism is structural: a state and an input letter may have several successors. 
This branching can be interpreted either existentially ---i.e., the input is accepted if some 
initial run on it is accepting--- or universally ---i.e., the input is accepted if all initial 
runs on it are accepting.
We use three-letter acronyms for classes of $\omega$-word automata. The first letter indicates 
the branching mode: deterministic (D), existential non-deterministic (N), or universal non-deterministic (U). 
The second letter indicates the acceptance condition: Büchi (B), coBüchi (C), or weak (W), where weak means 
that every strongly connected component is either contained in $F$ or disjoint from $F$. 
The third letter is $W$ for words, to disambiguate from tree automata. 
For example, $\UCW$ denotes universal coB\"uchi 
automata.
We denote by $\Lang(\mathcal{A})$ the language recognized by the automaton $\mathcal{A}$. 

An automaton $\mathcal{A}$ is \emph{counter-free} if for all states $q$, 
all finite words $w$, and all natural numbers $n$, 
if $\mathcal{A}$ admits a path on $w^n$ that starts and ends at 
$q$, then it also admits a path on $w$ that starts and ends at $q$. 
For a given automata class \textsc{ABC}, we denote by \textsc{ABC}$_{cf}$ its counter-free restriction. 




\subparagraph{Membership problems.}
Given two formalisms $F$ and $F'$, we will denote by $F \mapsto F'$ the membership problem from $F$ to $F'$, that is,
the problem of deciding, given an expression in $F$, if there is a language-equivalent expression in $F'$. 
For example, $\MSO \mapsto \WMSO$ asks whether a given \MSO sentence $\varphi$ can be transformed to an equivalent sentence $\varphi'$ of \WMSO. 
We will also denote by $F \cap F'$ the set of languages definable by both $F$ and $F'$. 

We recall that \APT stands for \emph{alternating parity tree automata}, 
which are as expressive as \MSO and the \emph{modal $\mu$-calculus} (\Lmu)
over full binary trees. 
\AWT is the subclass of \emph{alternating weak tree automata}, 
which captures \WMSO and the \emph{alternation-free $\mu$-calculus} (\AFLmu). 
For the definitions of \APT, \Lmu, and \AFLmu, we refer the reader to, e.g., \cite{Wil01},
while for \AWT and other variants we will recall their definitions at the beginning of Section \ref{automaton and conjecture}.

Below is a rephrasing of the celebrated Rabin's reduction theorem.

\begin{restatable}[{Rabin's reduction theorem \cite[Theorem 29]{rabin1970weakly}}]{theorem}{rabin}
\label{Rabin's theorem}
    Decidability of $\APT \membership \AWT$ implies decidability of $\MSO \membership \WMSO$ and $\Lmu \membership \AFLmu$. 
\end{restatable}

Although direct attacks on the latter logical membership problems have seen limited progress, 
the corresponding automata-theoretic problem has recently witnessed major advances~\cite{colcombet2008non,colcombet2013deciding,idir2025using}.
Through Rabin's reduction, these advances feed back into the logical setting, yielding progress on $\MSO \membership \WMSO$ and $\Lmu \membership \AFLmu$.

Our goal is to apply Rabin's reduction scheme within smaller source logics. 
Rather than taking the input specification from $\MSO$ or $\Lmu$, we consider 
specifications given in $\MPL$ or $\CTLs$. Rabin's reduction then immediately 
implies that decidability of $\APT \membership \AWT$ yields decidability of
$\MPL \membership \MPL \cap \WMSO$ and $\CTLs \membership \CTLs \cap \AFLmu$.
The key question is therefore to identify the target intersection: what is 
$\MPL \cap \WMSO$, or equivalently $\CTLs \cap \AFLmu$? 
In the sequel we mostly use the former notation, leaving the equivalence 
with the latter implicit. 
As argued in the introduction, this common fragment is a natural object 
of study, and the present paper makes progress towards its characterisation.




\section{The automaton and the conjecture} \label{automaton and conjecture}

We introduce the new class of \emph{counter-free hesitant weak $\omega$-tree automata},
by imposing a series of restrictions to alternating Büchi $\omega$-tree automata.
We then explain why this model is a natural candidate for capturing the common fragment $\MPL \cap \WMSO$. 
Finally, we discuss the perhaps surprising fact that $\MPL \cap \WMSO$ is strictly more expressive than $\FO$.

\subparagraph{Alternating B\"uchi $\omega$-tree automata.}
Given a set $X$, we think of its elements as \emph{atoms} and we denote by $\mathcal{B}^+(X)$ 
the set of positive Boolean formulae over $X$.
An \emph{alternating B\"uchi $\omega$-tree automaton} (\ABT) is a tuple 
$\mathcal{A} = \langle Q, \Sigma, \delta, q_I, F \rangle$, where $Q,\Sigma,q_I,F$ are as usual
and $\delta$ is a transition function from $Q \times \Sigma$ to $\mathcal{B}^+(\{\diamond, \square\} \times Q)$. 
Let $t$ be a $\Sigma$-tree.
A run of $\mathcal{A}$ on $t$ is an \emph{unranked}, \emph{unordered} tree, i.e., a tree in which 
every node has an arbitrary but finite set of children. It is also labelled over the infinite alphabet 
$Q\times\Dom(t)$, thus associating with each node a state of $\mathcal{A}$ and a corresponding 
node in $t$. Intuitively, a labelling $(q,v)$ of a node of a run represents a copy of the automaton 
that reads an input node $v$ with state $q$. Let $r$ be an unranked tree labeled over $Q\times\Dom(t)$. 
We define the satisfaction relation between a $(q,v)$-labelled node $u$ of $r$ and a positive Boolean 
formula in $\mathcal{B}^+(\{\diamond,\square\}\times Q)$ as follows:
\begin{itemize}
    \item $u \models (\diamond, q')$ if $v$ has a child $v'$ in $t$ and $u$ has a child $u'$ labelled by $(q',v')$;
    \item $u \models (\square, q')$ if for every child $v'$ of $v$ in $t$, there is a child $u'$ of $u$ labelled by $(q',v')$.
\end{itemize}
The satisfaction relation $\models$ is naturally extended to the Boolean constants $\top$ and $\bot$ and
to every positive Boolean combination of atoms. We then say that $r$ is a run on $t$ if:
\begin{enumerate}
    \item the root of $r$ is labelled by $(q_I, \epsilon_t)$, namely, by the initial state of $\mathcal{A}$ and by the root of $t$;
    \item for every $(q,v)$-labelled node $u$ of $r$, $u \models \delta(q, \sigma)$, where $\sigma$ is the label of $v$ in $t$.
\end{enumerate}
Concerning the above condition (2), we observe that the definition of the satisfaction
relation $\models$ between nodes and positive Boolean formulae requires the
existence
of a minimal set of children of each node $u$ of a run. On the other hand, it also
allows a node $u$ to contain spurious children that are not strictly needed to
satisfy condition (2). In particular, the property of being
a run of $t$ is preserved under the addition of subtrees that are themselves runs on
$t$, possibly starting at different nodes.
This turns out to be useful in order to avoid dealing with runs with leaves: indeed, even in the presence
of a transition of the form $\delta(q,\sigma)=\top$, 
we can still assume that
there
is an infinite path from $u$, e.g., ~one that repeats copies of $u$ itself
(note that every such copy of $u$ would still satisfy condition (2) above).
Finally, a run $r$ is \emph{accepting} if every maximal path in it 
satisfies the B\"uchi condition, i.e., every infinite path has infinitely 
many nodes labeled by some state in $F$. We say that a $\Sigma$-tree $t$ 
is accepted by $\mathcal{A}$ if $\mathcal{A}$ admits an accepting run on $t$. 

\subparagraph{Weakness.} \label{Weak condition}
An \ABT $\mathcal{A} = \langle Q, \Sigma, \delta, q_I, F \rangle$ is \emph{weak} (\AWT) 
if its state space can be partitioned into a sequence of components $Q_1, ..., Q_k$
in such a way that (1) for every $q\in Q_i$ and every $\sigma \in \Sigma$, 
$\delta(q, \sigma)$ can only contain atoms with states in $\bigcup_{j\le i} Q_j$,
and (2) for every $i$, either $Q_i\subseteq F$ or $Q_i\cap F=\emptyset$.
We call \emph{accepting} (resp., \emph{rejecting}) the components that
are contained in $F$ (resp., disjoint from $F$).


\subparagraph{Hesitancy.}
An \AWT $\mathcal{A} = \langle Q, \Sigma, \delta, q_I, F \rangle$ is \emph{hesitant} (\HWT) 
if every component $Q_i$ of the partition witnessing weakness satisfies one of the following 
additional conditions, for every $q,q'\in Q_i$ and $\sigma \in \Sigma$:
\begin{itemize}
    \item \emph{transient component}: $q'$ does not occur in any atom of $\delta(q, \sigma)$;
    \item \emph{existential component}: $q'$ can occur in the disjunctive normal form of $\delta (q, \sigma)$,
          but only in atoms of the form ($\diamond, q')$ and only disjunctively w.r.t.~to atoms with other states in $Q_i$;
    \item \emph{universal component}: $q'$ can occur in the conjunctive normal form of $\delta (q, \sigma)$,
          but only in atoms of the form ($\square, q')$ and only conjunctively w.r.t.~to atoms with other states in $Q_i$.
\end{itemize}

\subparagraph{Linearisation.}
Towards introducing the remaining restrictions, we need to describe a linearisation
operation that transforms any non-transient component of an \HWT $\mathcal{A}$ into an 
$\omega$-word automaton that reads annotated inputs. 
The idea is to simulate the path of a run of $\mathcal{A}$ that stays inside a given
component by a suitable $\omega$-word automaton, with either Büchi or coBüchi acceptance conditions
and non-determinism used either existentially or universally depending on the type of component, 
and at the same time encode the subruns of $\mathcal{A}$ that branch off the path by a suitable 
annotation of the input.

Let $\mathcal{A} = \langle Q, \Sigma, \delta, q_I, F\rangle$ be an \HWT
and let $q$ be a state that belongs to a non-transient component $Q_i$, 
according to some partition witnessing hesitancy of $\mathcal{A}$.
Further let $B$ be the set of atoms from $\{\diamond, \square\} \times (Q\setminus Q_i)$, 
namely, with states outside the component $Q_i$.
The \emph{linearisation} of $\mathcal{A}$ from $q$ is defined as the 
$\omega$-word automaton $\mathcal{A}^q_{\mathrm{exit}}$ obtained from $\mathcal{A}$ 
by restricting the state set to $Q_i$, adding a fresh state $q_{\mathrm{exit}}$,
declaring $q$ to be the new initial state,
annotating the input letters with subsets $C$ of $B$,
and redefining the transition function and the acceptance condition
by a case distinction, as follows.
\begin{itemize}
    \item If $Q_i$ is an existential component, then
          \begin{enumerate}
            \item for all $q'\in Q_i\cup\{q_{\mathrm{exit}}\}$, 
                  $\delta'(q',(\sigma,C))$ contains $q_{\mathrm{exit}}$ 
                  iff $q' = q_{\mathrm{exit}}$, or $\delta(q',\sigma) = \bot$, or
                  the disjunctive normal form of $\delta(q',\sigma)$ contains 
                  a clause of the form $\bigwedge C$;
            \item for all $q',q''\in Q_i$,
                  $\delta'(q',(\sigma,C))$ contains $q''$ iff
                  the disjunctive normal form of $\delta(q',\sigma)$ contains
                  a clause of the form $(\diamond, q'') \land \bigwedge C$.
          \end{enumerate}
          In addition, if $Q_i$ is accepting (resp., rejecting), then
          $\mathcal{A}^q_{\mathrm{exit}}$ is used as an $\NBW$ (resp., $\NCW$),
          with $Q_i$ as set of final states.
    \item If $Q_i$ is a universal component, then
          \begin{enumerate}
            \item for all $q'\in Q_i\cup\{q_{\mathrm{exit}}\}$, 
                  $\delta'(q',(\sigma,C))$ contains $q_{\mathrm{exit}}$ 
                  iff $q' = q_{\mathrm{exit}}$, or $\delta(q',\sigma) = \top$, or
                  the conjunctive normal form of $\delta(q',\sigma)$ contains 
                  a clause of the form $\bigvee C$;
            \item for all $q',q''\in Q_i$,
                  $\delta'(q',(\sigma,C))$ contains $q''$ iff
                  the conjunctive normal form of $\delta(q',\sigma)$ contains
                  a clause of the form $(\square, q'') \vee \bigvee C$.
          \end{enumerate}
          In addition, if $Q_i$ is accepting (resp., rejecting), then
          $\mathcal{A}^q_{\mathrm{exit}}$ is used as an $\UBW$ (resp., $\UCW$),
          with $Q_i$ as set of final states.
\end{itemize}

\subparagraph{Visibility.}
Let $\mathcal{A} = \langle Q,\Sigma,\delta,q_I,F\rangle$ be an \HWT, 
and let $\theta\in\{\diamond,\square\} \times Q$.
We write $\mathcal{A}^\theta$ for the automaton obtained from $\mathcal{A}$ 
by introducing a fresh initial state and by forcing the first transition to be $\theta$, 
independently of the label of the root of the input tree.
We say that a non-transient component $Q_i$ of $\mathcal{A}$ is \emph{visible} 
if for every pair of distinct annotations $C,C'\subseteq B$
that appear in some transitions of the linearisation $\mathcal{A}^q_{\mathrm{exit}}$, for any $q\in Q_i$, 
there are atoms $\theta\in C$ and $\theta'\in C'$ for which the $\omega$-tree languages
$\Lang(\mathcal{A}^\theta)$ and $\Lang(\mathcal{A}^{\theta'})$ are complement of each other. 
Accordingly, we say that $\mathcal{A}$ is \emph{visible} if every non-transient component is visible.

\subparagraph{Counter-freeness.} 
Recall the definition of counter-free $\omega$-word automaton given 
in the preliminaries.
The last restriction that we enforce requires the linearisation of every
non-transient component of a visible \HWT to be counter-free.
Accordingly, we call \emph{counter-free hesitant weak $\omega$-tree automata}
(\HWTcf for short) a Büchi $\omega$-tree automaton that is weak, hesitant, visible, and counter-free.

We remark that counter-freeness only makes sense when the automaton is visible.
Intuitively, this is because visibility guarantees that the annotations used in the linearisation 
of a component reflect genuine branching alternatives. In other words, distinct annotations are 
distinguishable by incompatible behaviours outside the component; hence they cannot serve as hidden 
markers that encode counting information while leaving the linearised $\omega$-word language counter-free.
On the other hand, it can be shown that a counter-free \HWT which is not visible can recognize 
regular $\omega$-tree languages that are not even definable in \MPL, see, e.g., \cite[Example 5.1]{BBMP24}. 


\subparagraph{The conjecture.} 
In \cite[Theorem 1.1]{MSS92} \WMSO was shown to be as expressive as \AWT,
and in \cite[Theorems 5.8 and 5.9]{BBMP24} 
\MPL was shown to be as expressive as counter-free, visible, hesitant $\omega$-tree automata
(\HTcf for short). Because \HTcf employ an acceptance condition that is stronger than B\"uchi but 
weaker than Parity, we know that 
$\HWTcf \subseteq \HTcf \cap \AWT \subseteq \MPL \cap \WMSO$.
Our Conjecture \ref{con:comfrg} claims that $\HWTcf$ are in fact as expressive as $\MPL \cap \WMSO$. 

The findings in the subsequent sections should provide evidence for the conjecture. 
The first intuition for its validity relies on the fact that 
\MPL restricts the \emph{structure} of quantified sets, while 
\WMSO restricts their \emph{cardinality},
and the two constraints seem to be orthogonal. 
More precisely, orthogonality should be witnessed at the syntactic level,
when we translate \MPL and \WMSO to their automata counterparts, \HTcf and \AWT.
If our conjecture were false, there would exist, e.g., an $\AWT$ that recognises an 
\MPL-definable language, yet structurally violates the constraints required by \HTcf. 
Given the independent nature of the syntactic constraints posed by the two classes of automata,
i.e., hesitancy and counter-freeness on one side and weakness on the other side,
we belive this is unlikely to happen.

\subparagraph{\FO is \WMPL but not $\MPL \cap \WMSO$.} \label{FO diverso}
We conclude the section by showing that the common fragment $\MPL \cap \WMSO$ is non-trivial. 
We first note that \WMPL, the fragment of \MSO in which monadic quantification is restricted 
to finite paths, is as expressive as \FO: indeed, every finite path in a tree is uniquely 
determined by its endpoints.
One might also expect \WMPL to coincide with $\MPL \cap \WMSO$, since, as argued earlier, 
\MPL and \WMSO impose orthogonal restrictions on monadic quantification. 
However, these restrictions interact semantically in a subtler way, 
and the expectation is in fact false:

\begin{restatable}{lemma}{fodiverso} \label{MPL, WMSO but not FO}
    There is an $\omega$-tree language $T$ definable in $\MPL \cap \WMSO$ but not in \FO.
\end{restatable}

\begin{proof} 
    Consider the language $T$ of $\omega$-trees where every infinite path from the root contains a node carrying the monadic predicate $p$.
    The underlying alphabet is $\Sigma=2^{\{p\}}$. This language is definable in \MPL, by the formula 
    $\neg \exists X \forall y (y \in X \rightarrow \neg p(y))$. It is equally definable in \WMSO, by the formula 
    $\exists X (\mathrm{root} \in X \land \forall y (y \in X \rightarrow (p(y) \lor \forall z(\mathrm{child}(z, y) \rightarrow z \in X))))$, where 
    $\mathrm{root} \in X = \exists x \forall y (x \leq y \land x \in X)$ and 
    $\mathrm{child}(z, y) = z > y \land \neg \exists w (z > w > y)$. 
    Finally, it was shown in \cite[Theorem 3]{BBMP23}, that this language 
    is not definable in \FO. 
\end{proof}

This shows that $\MPL \cap \WMSO$ is a genuine common fragment, not another presentation of $\FO$. 
It also suggests that the membership problem 
$\MPL \membership \FO$ cannot be settled by a Rabin-style reduction to $\APT \membership \AWT$. 
Although $\MPL = \HTcf \subseteq \APT$, reducing a given $\HTcf$ to an equivalent $\AWT$ 
would only show that the corresponding $\MPL$ formula is equivalent to some $\WMSO$ formula;
it would not give \FO-definability. 
Thus, even for a logic close to $\FO$ such as $\MPL$, \FO-definability remains an elusive question, 
not even reducible to the (still open) problem $\APT \membership \AWT$.

\section{\PCTL is equivalent to \HWTcf} \label{PCTL = automaton}

Recall that the witness language of Lemma~\ref{MPL, WMSO but not FO} is
over the alphabet $\Sigma=2^{\{p\}}$ and of the form 
$T = \{ t \mid \text{every infinite path from the root of $t$ carries $p$}\}$.
The following lemma shows that $T$ is recognised by an $\HWTcf$,
and also serves as a worked example of an $\HWTcf$.
Knowing that \FO is equivalent to a subclass of \HWTcf, as proved in \cite{BDMMP25},
this also implies that \FO is strictly less expressive than \HWTcf.
In particular, $T$ cannot be used to disprove our conjecture. 

\begin{lemma} \label{T expressible in HWT}
    The language $T$ from Lemma \ref{MPL, WMSO but not FO}
    is recognized by an \HWTcf.
\end{lemma}

\begin{proof}
    Let $\mathcal{A}_T = \langle \{q_I\}, 2^{\{p\}}, \delta, q_I, \emptyset \rangle$, 
    where 
    $\delta(q_I, \{p\}) = \top$ and 
    $\delta(q_I, \emptyset) = (\square, q_I)$. 
    This is an \APT that clearly recognises $T$, 
    since it accepts a tree iff it finds a $\{p\}$-labelled node along 
    every infinite path from the root. 
    We claim that $\mathcal{A}_T$ is also an \HWTcf.
    Observe that 
    $\mathcal{A}_T$ is weak, with only one 
    component $Q_1 = \{q_I\}$ that is rejecting and universal.
    The linearisation of $\mathcal{A}_T$ from $q_I$ is an \UCW
    of the form 
    $\mathcal{A}^{q_I}_{\mathrm{exit}} = \langle \{q_I, q_{exit}\}, \Sigma', \delta', q_I, \{q_I\}\rangle$. 
    Note that we can assume the annotated alphabet to be $\Sigma' = \Sigma \times \{\emptyset\}$, 
    since no atom outside the component occurs in the image of $\delta$.
    The transition function $\delta': \{q_I,q_{\mathrm{exit}}\} \times \Sigma' \rightarrow 2^{\{q_I,q_{\mathrm{exit}}\}}$ is such that 
    $\delta'(q_I,(\emptyset,\emptyset))=\{q_I\}$ and 
    $\delta'(q,(\sigma,\emptyset))=\{q_{\mathrm{exit}}\}$ in all remaining cases,
    i.e., when $q=q_{\mathrm{exit}}$ or $\sigma=\{p\}$.
   It is easy to see that this \UCW is visible and counter-free: indeed,
   there is only one possible annotation that appears in the transitions, i.e., $C=\emptyset$, 
   and the transition graph consists solely of self-loops and a transition from $q_I$ to $q_{\mathrm{exit}}$. 
   $\mathcal{A}_T$ is therefore an \HWTcf.
\end{proof}

We now turn to a logical characterisation of $\HWTcf$ in terms of $\PCTL$. 
Towards the end of the section, we will see that, conditional on Conjecture~\ref{con:comfrg}, 
this characterisation yields a Rabin-style reduction from the membership problem $\CTLs \membership \PCTL$ 
to the automata-theoretic problem $\APT \membership \AWT$.
To prove the characterisation, we introduce two intermediate results: 
a normal form for $\PCTL$ and a syntactic fragment of $\CTLs$ that is equivalent to $\PCTL$.


\begin{restatable}{lemma}{pctlnorm} \label{PCTL normal form}
For every \PCTL formula, there is an equivalent one generated by:
\begin{align*}
\varphi ~::= &~
  p
  ~\bigm|~
  \neg p
  ~\bigm|~
  \varphi \lor \varphi
  ~\bigm|~
  \varphi \land \varphi
  ~\bigm|~
  \E\F \psi
  ~\bigm|~
  \E\G \psi
  ~\bigm|~
  \A\F \psi
  ~\bigm|~
  \A\G \psi
\\
\psi ~::= &~
  \varphi
  ~\bigm|~
  \psi \lor \psi
  ~\bigm|~
  \psi \land \psi
  ~\bigm|~
  \Y \psi
  ~\bigm|~
  \tilde{\Y} \psi
  ~\bigm|~
  \psi \s \psi
\end{align*}
\end{restatable}

\begin{proof} [Proof sketch]
    The proof is by structural induction. 
    The path formulae of $\PCTL$ belong to the fragments of $\PLTL$ in negation normal form that forbid the use of either $\R$ or $\U$.
    Thus one can use 
    Proposition \ref{safeLTL with past is equal to safeLTL} to turn those formulae into 
    equivalent ones from \safePLTL and \cosafePLTL. This is
    consistent with the above grammar. 
\end{proof}

The above normal form
clearly shows that \PCTL expressive power is deeply linked to the ability to express 
both \emph{safety} and \emph{co-safety} properties along paths. Moreover, it gives us a way to identify a syntactic fragment of \CTLs equivalent to \PCTL. This is useful for two reasons: first, it shows how to retain \PCTL expressive power while dropping past operators; second, it will be instrumental in the proof that \PCTL is expressively equivalent to \HWTcf.

\begin{restatable}{lemma}{cosafeCTL} \label{PCTL as CTLs fragment}
    \PCTL is as expressive as \obbligationCTLs, generated by the grammar
 \begin{align*}
\varphi ~::= &~
  p
  ~\bigm|~
  \neg p
  ~\bigm|~
  \varphi \lor \varphi
  ~\bigm|~
  \varphi \land \varphi
  ~\bigm|~
  \E \psi
  ~\bigm|~
  \E \chi
  ~\bigm|~
  \A \psi
  ~\bigm|~
  \A \chi
\\
\psi ~::= &~
  \varphi
  ~\bigm|~
  \psi \lor \psi
  ~\bigm|~
  \psi \land \psi
  ~\bigm|~
  \X \psi
  ~\bigm|~
  \psi \U \psi
\\
\chi ~::= &~
  \varphi
  ~\bigm|~
  \chi \lor \chi
  ~\bigm|~
  \chi \land \chi
  ~\bigm|~
  \X \chi
  ~\bigm|~
  \chi \R \chi
\end{align*}
\end{restatable}

\begin{proof} [Proof sketch]
    The proof is again by induction in both directions. 
    For the non-trivial cases, the idea is to exploit the fact that path formulae of the \PCTL normal form and those of \obbligationCTLs are expressively equivalent, thanks to Theorem \ref{safeLTL equivalent to safePLTL}. 
    \qedhere 

\end{proof}

Now, we are ready to establish the equivalence between \PCTL and \HWTcf. 
We say that an \HWTcf is \emph{universal} (resp., \emph{existential}) if
its transition function only produces conjunctions (resp., disjunctions)
of atoms of the form $(\square, q)$ (resp., $(\diamond, q)$).
Similarly, we define \APCTL (resp., \EPCTL) as the fragment of \PCTL 
in which only universal (resp., existential) path quantifiers are allowed.

\autopctl*

\begin{proof} [Proof sketch]
The proof relies on using linearisations of non-transient components 
of \HWTcf in one direction, and \obbligationCTLs in the other direction. 
To translate an \HWTcf to an equivalent \PCTL formula, we adapt the proof 
of \cite[Theorem 12]{BDMMP25}, proceeding by induction on the components' order. 
The transient case is straightforward; the case of a universal (resp., existential) 
component is addressed by linearising the component and translating it into an
equivalent \safeLTL (resp., \cosafeLTL) formula. Recombining the latter formulas 
yields an \obbligationCTLs formula equivalent to the initial \HWTcf. 
Finally, Lemma~\ref{PCTL as CTLs fragment} provides an equivalent \PCTL formula. 
This also proves the second statement of the theorem, since translating a universal 
(resp., existential) \HWTcf never requires existential (resp., universal) path quantification. 

For the converse direction, we apply again Lemma~\ref{PCTL as CTLs fragment} 
to transform the input \PCTL formula into an equivalent \obbligationCTLs formula. 
To enforce visibility in the target automaton, we rely on an expressively complete 
syntactic variant of \obbligationCTLs~\cite[Section 5.1]{BBMP24} designed to isolate 
path quantifiers, ensuring that every pair of consecutive path quantifiers is 
explicitly separated. Then, the translation is fairly routine: the only non-trivial 
case is the path quantifier, which is handled as in, e.g., \cite[Theorem 5.9]{BBMP24}.
\end{proof}



We can now state a Rabin-style reduction theorem, conditional on Conjecture~\ref{con:comfrg}: 

\mainthm*



\section{The common fragment of \PCTL and \LTL is decidable} \label{Common fragment}

In this section we show another major application of the logical characterization 
of \HWTcf given in Theorem \ref{thm:hwtcf&pctl}: specifically, we establish the 
decidability of the common fragment of \LTL and \PCTL. 
More precisely, we prove that both membership problems 
$\PCTL \membership \LTL$ and $\LTL \membership \PCTL$ are decidable, and therefore, 
$\MSO \membership \LTL \cap \PCTL$ is decidable too. 
This result is interesting for at least two reasons: first, it provides
decidability results for membership problems of tree logics, which are
notoriously among the hardest problems in language theory; second, it offers a
new way to attack the almost 40-year-old open problem of deciding the common
fragment of \LTL and \CTL.

Given a regular $\omega$-word language $L$, let $\triangle [L]$ be the regular $\omega$-tree 
language of all trees such that every infinite path starting at the root is in $L$.

\begin{restatable}{lemma}{Wordtotree} \label{From words to trees}
    Let $\mathcal{A}$ be a \DBWcf. One can construct a universal \HWTcf recognizing $\triangle[\Lang(\mathcal{A})]$.
\end{restatable}

\begin{proof}[Proof sketch]
    Given a \DBWcf $\mathcal{A} = \langle Q, \Sigma, \delta, q_I, F \rangle$, 
    we first construct an equivalent \UWW $\mathcal{B}$ by taking two disjoint copies of $\mathcal{A}$.
    Intuitively, the first copy follows the unique run of $\mathcal{A}$ forever, and it is accepting. 
    Whenever this run enters a non-final state, the universal automaton $\mathcal{B}$ spawns a branch 
    in the second copy, which checks the local obligation that a final state is reached later. 
    Since the second copy is rejecting, such a branch cannot remain there forever; 
    it is discharged exactly when the run reaches a final state. 
    Formally, we let 
    $\mathcal{B} = \langle (Q\times \{1\}) \cup ((Q \setminus F) \times\{2\}), \Sigma, \delta', (q_I, 1), (Q\times\{1\}) \rangle$,
    where $\delta'$ is defined as follows:   
    \begin{itemize}
    \item if $\delta(q, \sigma) = q' \notin F$, then $\delta'((q, 1), \sigma$) = $\{(q',1), (q', 2)\}$ and $\delta'((q,2), \sigma) = \{(q', 2)\}$;

    \item if $\delta(q, \sigma) = q' \in F$, then $\delta' ((q,1), \sigma) = \{(q', 1)\}$ and $\delta'((q,2), \sigma) = \emptyset$.
    \end{itemize}
    Moreover, $\mathcal{B}$ is weak: the first component consists only of accepting states, 
    the second one only of rejecting states, and there are no transitions from the second 
    component back to the first. 
    Counter-freeness is preserved as well, since every cycle in $\mathcal{B}$ projects 
    to a cycle of $\mathcal{A}$.

    The second step turns $\mathcal{B}$ into a tree automaton. 
    For each transition of $\mathcal{B}$, every successor state $(q,i)$, for $i\in\{1,2\}$,
    is replaced by the atom $(\square,(q,i))$. 
    In this way, the word automaton is simulated uniformly along all branches of the input tree. 
    The resulting automaton is a universal $\HWTcf$ and recognises precisely $\triangle[\Lang(\mathcal{A})]$.
\end{proof}

The following theorem gives an automaton-theoretic characterization of a robust class of LTL-definable
$\omega$-tree languages, and is crucial for the main decidability result presented in this section
(Theorem \ref{thm:memprb}).

\begin{theorem} \label{main theorem}
    For an \LTL-definable $\omega$-word language $L$, the following are equivalent:
    (i) $L$ is recognized by a \DBWcf,
    (ii) $\triangle[L]$ is recognized by a universal \HWTcf, 
    (iii) $\triangle[L]$ is recognized by an \HWTcf, 
    (iv) $\triangle[L]$ is recognized by an \AWT.
\end{theorem}

\begin{proof}
Assume first that $L$ is recognised by a $\DBWcf$. 
By Lemma~\ref{From words to trees}, $\triangle[L]$ is recognised by a universal $\HWTcf$, so (i) implies (ii). 
The implications from (ii) to (iii) and from (iii) to (iv) are immediate, since universal $\HWTcf$ are $\HWTcf$, and $\HWTcf$ are $\AWT$.
It remains to prove the implication from (iv) to (i). Suppose that $\triangle[L]$ is recognised by an $\AWT$. 
By~\cite[Theorem~1.1]{MSS92}, it is also recognised by a \emph{non-deterministic B\"uchi tree automaton}. 
Then, by~\cite[Theorem~3.1]{kupferman2006relating}, $L$ is recognised by a $\DBW$. Since $L$ is $\LTL$-definable, 
it is star-free; hence, after minimising the $\DBW$, we obtain an equivalent $\DBWcf$~\cite{perrin1995semigroups}.
\end{proof}

Theorem \ref{main theorem}
also yields the following logical equivalences (see again Figure \ref{fig:venndiag}):

\comfrg*

\begin{proof}
Let $L$ be an $\LTL$-definable $\omega$-word language. 
It is enough to show that definability of $\triangle[L]$ in $\APCTL$, in $\PCTL$, 
and in $\AFLmu$ are all equivalent to recognisability of $L$ by a $\DBWcf$.
If $\triangle[L]$ is definable in $\PCTL$, then by Theorem~\ref{thm:hwtcf&pctl} it 
is recognised by an $\HWTcf$, and Theorem~\ref{main theorem} yields a $\DBWcf$ for $L$. 
The case of $\APCTL$ follows since $\APCTL\subseteq\PCTL$. 
If $\triangle[L]$ is definable in $\AFLmu$, then it is recognised by an $\AWT$, thanks to \cite{KV05}; 
again Theorem~\ref{main theorem} yields a $\DBWcf$ for $L$.
Conversely, if $L$ is recognised by a $\DBWcf$, then Theorem~\ref{main theorem} gives a universal $\HWTcf$ recognising $\triangle[L]$. 
By Theorem~\ref{thm:hwtcf&pctl}, this yields an equivalent $\APCTL$ formula. 
Definability in $\PCTL$ follows from $\APCTL\subseteq\PCTL$, and definability in $\AFLmu$ follows since $\HWTcf$ are $\AWT$, which characterise $\AFLmu$~\cite{KV05}.
%
\end{proof}

The above result might appear surprising, since 
Boja{\'n}czyk proved in \cite{bojanczyk2008common}  that \LTL $\cap$ \ACTL $\subsetneq$ \LTL $\cap$ \CTL. Indeed, Theorem \ref{thm:comfrg} implies that it is possible to avoid existential path quantification and capture the common fragment of \LTL and \PCTL, employing past operators instead. To exemplify this idea, we show how to define in \APCTL the language used by Boja{\'n}czyk to show the strict inclusion already mentioned.

\begin{proposition}
    The regular $\omega$-tree language $\triangle[(ab)^*a(ab)^*c^\omega]$ is definable in \APCTL.
\end{proposition}

\begin{proof}
    The formula that defines the $\{a,b,c\}$-tree language $\triangle[(ab)^*a(ab)^*c^\omega]$ is as follows, 
    where $\PP \psi$ is a shorthand for $\top\s\psi$:
     \[
    \begin{array}[b]{c}
        a \land \A\F c \land \A\G (\neg(b \land \Y b) \land
        \neg (a \land \Y a \land \Y\PP(a \land \Y a)) \land (\Y c \rightarrow c) \land \\
        ((c \land \Y\neg c) \rightarrow \Y((a \land \neg \PP (a \land \Y a)) \lor (b \land 
        \PP(a \land \Y a)))
    \end{array} \qedhere
    \]
\end{proof}

In the formula above, past 
operators 
are 
necessary to capture the complex 
alternation of $a$, $b$, and $c$ without employing the existential 
path quantifier $\E$ while 
respecting the syntactic 
constraints of $\CTL$.
This unveils an unexpected power of past operators in the context of tree temporal logics. We finally conclude with our main decidability results.

\memprb*

\begin{proof}
  By Theorem \ref{main theorem}, given an \LTL formula $\varphi$, it suffices to
  decide if $\Lang(\varphi)$ is recognizable by a \DBW (the counter-free
  requirement is trivially fulfilled, since 
  any \LTL definable language must be recognizable by a counter-free
  automaton).
  By \cite[Theorem 4.1, Theorem 5.1] {KV05}, deciding \LTL
  $\mapsto$ \DBW is in EXPSPACE and PSPACE-hard.

  For \PCTL $\mapsto$ \LTL, we can always rewrite a \PCTL formula into \obbligationCTLs by Lemma \ref{PCTL as CTLs fragment}, and then apply the algorithm from \cite{CD88} to decide if there is an equivalent \LTL formula.
\end{proof}

The first statement of the above theorem implies the following. 

\ltlctl*

Thanks to the above corollary, the 40-year-old open problem can be attacked along a new direction: determine which \PCTL formulae are equivalent to some \CTL formula. This is a promising new approach, since some subcases are already solved (see, e.g., \cite{laroussinie1995hierarchy}). To conclude, we can settle that the common fragment \LTL $\cap$ \PCTL is decidable.

\begin{restatable}{corollary}{msocmn} \label{cor: msocmn}
The membership problem \MSO $\mapsto$ \LTL $\cap$ \PCTL is decidable. 
\end{restatable}




%




\section{No known fragment of \WMSO is equivalent to \PCTL} \label{no current}

In this section we will show that 
no known fragment of \WMSO is equivalent to \PCTL. In other words, to understand \PCTL expressive power with respect to \WMSO, one needs to define a new fragment of \WMSO. Indeed, given Theorem \ref{thm:rabthm}, one could wonder what is the exact fragment of \WMSO equivalent to \MPL $\cap$ \WMSO. We have shown that, assuming Conjecture \ref{con:comfrg} to be true, \MPL $\cap$ \WMSO is \PCTL. Nevertheless, it is a natural question to ask for a fragment of \WMSO, rather than for a temporal logic. In this section we initiate this work, by isolating the lower and upper bounds of \PCTL expressive power with regard to \WMSO fragments, to pave the way for future work on this matter.

To the best of our knowledge, there are three defined semantic fragments of 
\WMSO: \FO (which is equivalent to \WMPL, as 
already argued), \emph{Weak Chain Logic} (\WCL) 
\cite{Tho84} and \emph{Weak Tree Logic} (\WMTL) 
\cite{BBMP23}. While \FO and \WMTL have already been defined, we briefly recall the definition of \WCL. In \WCL, second order quantification is restricted to \emph{finite chains}, that is, finite sets of nodes that are pairwise comparable under the transitive closure of the tree's edge relation.

First, while it is already established that \FO is strictly less expressive than \PCTL \cite{BDMMP25}, this fact also emerges as a direct consequence of Lemma \ref{T expressible in HWT} and Theorem \ref{thm:hwtcf&pctl}. Furthermore, \WCL is known to be incomparable with \MPL \cite{BL13}. Because \MPL strictly subsumes \PCTL, this incomparability necessarily implies that \PCTL cannot be equivalent to \WCL. 
Thus, \WMTL remains the only reasonable fragment of \WMSO to investigate. Recall that in \WMTL set quantification ranges over sets of nodes that are finite subtrees. Although in \cite{BBMP23} the authors explored the expressive power of this logic, the precise relationship between \WMTL and \MPL in terms of expressiveness remained an open problem. Here, we show that these two logics are incomparable, implying also that \PCTL is strictly less expressive than \WMTL. The following is obtained by an easy translation by structural induction on \PCTL formulae. 

\begin{restatable}{lemma}{pctlwmtl} \label{PCTL and WMTL}
    \PCTL is at most as expressive as \WMTL.
\end{restatable}

Since there is no defined intermediate \WMSO fragment between \FO and \WMTL, one could believe \PCTL to be equivalent to \WMTL. However, the following theorem shows this is not the case, while also closing the open question by \cite{BBMP23}.

\wmtlmpl*

\begin{proof}[Proof sketch]

    The difficult direction is \WMTL $\not\subseteq$ \MPL. Consider the alphabet $2^{\{p,q\}}$. The separating language will be $U =$ ``there is a finite subtree containing the root such that: its leaves are labeled $\{p,q\}$, its internal nodes labeled $\{q\}$ have at least one child and its internal nodes not labeled $\{q\}$ have exactly two children''. This language is definable in \WMTL:
    \[
        \exists X. root \in X \land \forall z \in X \rightarrow \land
        \begin{cases}
           (\neg p(z) \rightarrow \exists y(child(z,y) \land y \in X)) \\
          (\neg q(z) \rightarrow \forall y(child(z,y) \rightarrow y \in X))
        \end{cases}
    \]
    Our claim is that this language is not \MPL definable. We show it exploiting two families of trees.
    First, let a \emph{$\Sigma$-multicontext} be a finite or infinite $\Sigma \cup \{*_1, ..., *_k\}$-tree, for $k \in \mathbb{N}$, in which every internal node is labeled by a symbol of $\Sigma$ and every leaf (if any) is labeled by $*_i$, for some $i \in [1,k]$. Consider the following $2^{\{p,q\}}$-multicontext $c$: 
   
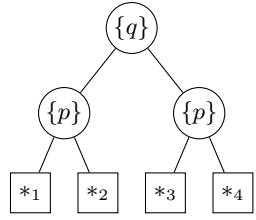
  \begin{figure}[h!]
  \centering
  \begin{forest}
    for tree={
      draw,
      circle,         
      minimum size=1.5em, 
      inner sep=1pt,      
      l=2em,              
      s sep=1em,          
      font=\small         
    },
    where n children=0{shape=rectangle}{}
    [$\{q\}$
      [$\{p\}$ 
        [$\ast_1$] 
        [$\ast_2$] 
      ]
      [$\{p\}$ 
        [$\ast_3$] 
        [$\ast_4$] 
      ]
    ]
  \end{forest}
\end{figure}
Let $l$ (resp., $r$) be the $2^{\{p,q\}}$-tree whose root is labeled by $\{q\}$ (resp., $\{p,q\}$) and all its other nodes are labeled by $\emptyset$. We define $t_1$ to be $c$ in which the leaf labeled $*_1$ is substituted by $l$ and leaves labeled $*_2, *_3$ and $*_4$ are substituted by $r$, while $t'_1$ is defined as $c$ in which leaves labeled $*_1,*_3$ are substituted by $l$ while nodes labeled $*_2, *_4$ are substituted by $r$. Then, $t_1 \in U$ and $t'_1 \notin U$. We then define infinite such trees as follows: for every $i \ge 2$, $t_i$ is $c$ in which the leaf labeled $*_1$ is substituted by $t'_{i-1}$ and leaves labeled $*_2, *_3$ and $*_4$ are substituted by $t_{i-1}$, while $t'_i$ is $c$ in which leaves labeled $*_1$ and $*_3$ are substituted by $t'_{i-1}$ while leaves labeled $*_2$ and $*_4$ are substituted by $t_{i-1}$. Again, for every $i \ge 2$, $t_i \in U$ and $t'_i \notin U$. 
We show by induction on \CTLs formulae that, for every $n \in \mathbb{N}$, \CTLs formulae of size at most $n$, where their size is defined as the length of the strings for representing the formulae, cannot distinguish $t_n$ and $t'_n$. This implies that \CTLs cannot define $U$, yielding the theorem.
\end{proof}

This theorem shows that the branching power of \WMTL is in a strong sense more expressive than just path quantification. 
Our result also shows that \WMTL $\neq$ \MPL $\cap$ \WMSO, further solidifying our conjecture. Indeed, we can now see that if our conjecture is false, there would be not one but even \emph{two} ``meaningful'' logics strictly in between \FO and \WMTL, which are already fairly close formalisms with regard to expressiveness. This seems unlikely, since it is arguably already surprising the existence of one such a logic.


\section{Conclusions}

We have proved a Rabin-style theorem for the \MSO fragment called \MPL and the temporal logic \CTLs (Theorem
\ref{thm:rabthm}), assuming a conjecture (Conjecture \ref{con:comfrg})
regarding the expressive power of the class of \HWTcf. We
have shown that this class of automata is equivalent to the temporal logic \PCTL
(Theorem \ref{thm:hwtcf&pctl}) and to the syntactic fragment of \CTLs called
\obbligationCTLs (Lemma \ref{PCTL as CTLs fragment}). Moreover, we have shown that
the membership problems \LTL $\mapsto$ \PCTL and \PCTL $\mapsto$ \LTL are
decidable (Theorem \ref{thm:memprb}), as well as \MSO $\mapsto$ \LTL $\cap$ \PCTL (Corollary \ref{cor: msocmn}),
implying that the almost 40-year-old open question regarding decidability of
\LTL $\mapsto$ \CTL can be reduced to the arguably easier problem of showing
decidability of \PCTL $\mapsto$ \CTL (Corollary \ref{cor:ltl2ctl}). Given
the progress already made in, e.g., \cite{laroussinie1995hierarchy}, this seems
promising. Finally, we have shown that the expressive power of \PCTL is strictly in
between that of \FO and \WMTL (\cite{BDMMP25}
and Theorem \ref{thm:wmtl2mpl}), implying that no known \WMSO fragment is
equivalent to \PCTL, leaving the definition of such a fragment as a natural open question. 
The obvious final open question regards the truth of
Conjecture~\ref{con:comfrg}.

\bibliography{references,referencesFM}
\appendix

\section{Further preliminaries}

\subparagraph{Temporal logics.}

\begin{definition} [Full \PCTLs semantics]
    The satisfaction relations $\models_s$ and $\models_p$, respectively between pointed trees and
state formulae and between path-pointed trees and path formulae, are defined by a mutual induction as follows:
\begin{itemize}
\item $(t,u) \models_s p$ iff $p \in t(u)$,
\item $(t,u) \models_s \neg p$ iff $p \notin t(u)$,
\item $(t,u) \models_s \neg\varphi$ iff $(t, u) \nvDash_s \varphi$,
\item $(t,u) \models_s \varphi_1 \lor \varphi_2$ iff $(t,u) \models_s \varphi_1$ or $(t,u) \models_s \varphi_2$,
\item $(t,u) \models_s \varphi_1 \land \varphi_2$ iff $(t,u) \models_s \varphi_1$ and $(t,u) \models_s \varphi_2$,
\item $(t,u) \models_s \E \psi$ iff there exist an infinite path $\pi$ starting at $u$ such that $(t,\pi,0) \models_p \psi$,
\item $(t,u) \models_s \A \psi$ iff for every infinite path $\pi$ starting at $u$, $(t,\pi,0) \models_p \psi$
\item $(t,\pi,i) \models_p \varphi$ iff $(t,\pi(i)) \models_s \varphi$,
\item $(t,\pi, i) \models_p \neg\psi$ iff $(t, \pi, i) \nvDash_p \psi$,
\item $(t,\pi, i) \models_p \psi_1 \lor \psi_2$ iff $(t,\pi, i) \models_p \psi_1$ or $(t,\pi, i) \models_p \psi_2$,
\item $(t,\pi, i) \models_p \psi_1 \land \psi_2$ iff $(t,\pi, i) \models_p \psi_1$ and $(t,\pi, i) \models_p \psi_2$,
\item $(t,\pi,i) \models_p \X \psi$ iff $(t,\pi,i+1) \models_p \psi$,
\item $(t,\pi,i) \models_p \Y \psi$ iff $i > 0$ and $(t,\pi,i-1) \models_p \psi$,
\item $(t,\pi,i) \models_p \tilde{\Y} \psi$ iff $i = 0$ or $(t,\pi,i-1) \models_p \psi$,
\item $(t,\pi,i) \models_p \F \psi$ iff there is $j \ge i$ such that $(t,\pi,j) \models_p \psi$,
\item $(t,\pi,i) \models_p \G \psi$ iff for every $j \ge i$, $(t, \pi, j) \models \psi$,
\item $(t,\pi,i) \models_p \psi_1 \U \psi_2$ iff there is $j \ge i$ such that $(t,\pi,j) \models_p \psi_2$ and for all $k=i,\dots,j-1$, $(t,\pi,k) \models_p \psi_1$,
\item $(t,\pi,i) \models_p \psi_1 \R \psi_2$ iff $(t,\pi,j) \models_p \psi_2$ holds for all $j \geq i$, or there is there is $j \ge i$ such that $(t,\pi,j) \models_p \psi_1 \land \psi_2$ and for all $k=i,\dots,j-1$, $(t,\pi,k) \models_p \psi_2$,
\item $(t,\pi,i) \models_p \psi_1 \s \psi_2$ iff there is $j \le i$ such that $(t,\pi,j) \models_p \psi_2$ and for all $k=j+1,\dots,i$, $(t,\pi,k) \models_p \psi_1$.
\end{itemize}
\end{definition}

We will use $\models$, since there is no possibilty of confusion.
\subparagraph{Linear time formalisms.}
By seeing $\omega$-words over $\Sigma$ as degenerate cases of $\Sigma$-trees, where every node has exactly
one successor, one can interpret \MSO sentences and $\PCTLs$ formulae directly over $\omega$-words.  Clearly, path quantifiers become redundant, and one can consider a logic consisting only of path formulae.
However, because in an $\omega$-word there is only one maximal path, the path quantifiers are of no use, and similarly for the Boolean connectives appearing in the grammar of a state formula.
In particular, one can define a fragment of $\PCTLs$ where the only state formulae allowed are the atomic propositions,
giving rise to a simpler grammar for the so-called \emph{Linear Temporal Logic with Past} (\PLTL for short):
\[
\psi ~::=~
  p
  ~\bigm|~
  \neg p
  ~\bigm|~
  \psi \lor \psi
  ~\bigm|~
  \psi \land \psi
  ~\bigm|~
  \X \psi
  ~\bigm|~
  \Y \psi
  ~\bigm|~
  \tilde{\Y} \psi
  ~\bigm|~
  \psi \U \psi
  ~\bigm|~
  \psi \R \psi
  ~\bigm|~
  \psi \s \psi
\]
We consider the following fragments of \PLTL:
\begin{itemize}
    \item \cosafeLTL (resp., \safeLTL) is \PLTL in which past operators and the future operator $\R$ (resp., $\U)$ are disallowed;
    \item \cosafePLTL (resp., \safePLTL) consists of \PLTL formulae of the form $\F\alpha$ (resp., $\G\alpha$), where $\alpha$ contains only past operators.
\end{itemize}
The pure future fragment of $\PLTL$, denoted $\LTL$, is obtained by further disallowing the past operators
$\tilde{\Y}, \Y$ and $\s$. The semantics of both $\PLTL$ and $\LTL$ is inherited from $\PCTLs$, and the definitions of equivalence and expressiveness are as before, with the only difference that in this setting
the models are restricted to be $\omega$-words.

\PsafeLTL*

\begin{proof}
We prove the first statement characterizing cosafety properties; the characterization of safety properties follows
similar arguments, with the only difference that one manipulates the operator $\R$ instead of $\U$.
Since $\cosafeLTL$ is clearly included in
the fragment of $\PLTL$ that consists of formulae in negation normal form and that forbids the use of $\R$, the only interesting direction is the translation from the latter fragment to the former one.
%
%
For this direction, it suffices to apply Gabbay's Separation Theorem \cite[Chapter 10]{gabbay2000temporal} (in particular, \cite[Lemma 10.2.3]{gabbay2000temporal}),
which allows one to transform any given \PLTL formula into a
\emph{strongly equivalent} Boolean combination of pure past and pure future formulae.
By ``strong equivalence'' here we mean full state equivalence, namely, equivalence w.r.t. every structure and every position in it.
A close inspection to the proof of the Separation Theorem shows that the transformation
is performed inductively on the syntax of the input formula. Moreover, by design, the
translation rules mantain the following invariant: if the input formula is in negation normal form
and does not use $\R$, then the translated formula is also in negation normal form and does not use
$\R$ either.
As a consequence, from an input \PLTL formula $\psi$ in negation normal form that does not use $\R$ we obtain a strongly equivalent
formula $\psi'$ in negation normal form, that does not use $\R$, and that is a positive Boolean combination of pure future and pure past formulae.
Also recall that a pure past formula of the form $\Y \psi$, evaluated at the
initial position of a structure, evaluates to false; instead, a pure past
formula of the form $\psi_1 \s \psi_2$, evaluated at the initial position of a
structure, evaluates to $\psi_2$.
So, by subtituting in $\psi'$ all pure past formulae of the form  $\Y \psi$ (resp., $\psi_1 \s \psi_2$) with $\bot$ (resp., $\psi_2$), with $\bot=p \land \neg p$, one obtains
a formula $\psi''$ of $\cosafeLTL$ that is initially equivalent to $\psi$. Hence, by Theorem \ref{safeLTL equivalent to safePLTL}, it follows that there is a \cosafePLTL formula equivalent to $\psi''$.
\end{proof}

\subparagraph{Word automata}
A \emph{non-deterministic B\"uchi word automaton} is a tuple
$\mathcal{A} = (Q, \Sigma, \delta, q_I, F)$,
where $Q$ is a finite set of states,
$\Sigma$ is a finite alphabet,
$\delta:  Q\times\Sigma\rightarrow 2^Q$
is a transition function,
$q_I\in Q$ is an initial state, and
$F \subseteq Q$ is a set of \emph{accepting states}.
A non-deterministic B\"uchi word automaton is \emph{deterministic} if
$\delta(q,a)$ is a singleton for every $q\in Q$ and $a\in\Sigma$.

The definition of a run $\rho$ is as usual. The B\"uchi condition states that a run $\rho$ is accepting if it visits \emph{infinitely often} states in $F$. The dual \emph{coB\"uchi} condition instead states that a run $\rho$ is accepting if it visits \emph{finitely often} states in $F$. The \emph{weak} condition was defined in Subsection \ref{Weak condition} for tree automata, and has an equivalent definition for word automata. 

In the case of non-deterministic automata, an infinite word $w$ is \emph{accepted} by $\mathcal{A}$ if there is an accepting run
of $\mathcal{A}$ on $w$. We also consider \emph{universal non-determinism} and call \emph{universal} an automaton employing it. Given a universal automaton $\mathcal{A}$, we say that an infinite word $w$ is \emph{accepted} by $\mathcal{A}$ if every run of $\mathcal{A}$ on $w$ is accepting. 

An automaton $\mathcal{A}$ is \emph{looping} if it is a tuple $\langle Q \uplus \{q_{sink}\}, \Sigma, \delta, q_I, Q \rangle$, such that $\delta(q_{sink}, \sigma) = q_{sink}$, for every $\sigma \in \Sigma$. Notice that the definition holds both for B\"uchi and coB\"uchi automata.

\begin{proposition} \cite[Lemma 11]{BDMMP25} \cite[Theorem 1]{alpern1987recognizing}  \label{looping automata and results}
    Looping \NBWcf, looping \UBWcf, \safeLTL and \safePLTL are equivalent formalisms. The same holds for looping \NCWcf, looping \UCWcf, \cosafeLTL and \cosafePLTL.
\end{proposition}


\subparagraph{Membership problems.}

\rabin*

\begin{remark}
    Rabin's Theorem actually states that an regular $\omega$-tree language $T$ is \WMSO definable iff both $T$ and the complement of $T$ are recognized by a \emph{non-deterministic B\"uchi tree automaton} (\NBT). Hence, decidability of \APT $\mapsto$ \NBT implies decidability of \MSO $\mapsto$ \WMSO. Nevertheless, by \cite[Theorem 1.1]{MSS92}, $T$ and its complement are recognized by an \NBT iff $T$ is recognized by an \AWT. Combining the two results one gets the form of Rabin's theorem we stated. The \Lmu final part follows from \cite{JW96} and \cite{AN92}.
\end{remark}


\section{Proofs from Section \ref{automaton and conjecture}}

\begin{proposition}
    \FO and \WMPL are equivalent formalisms.
\end{proposition}

\begin{proof}
   Given a \WMPL formula $\psi$, one can obtain an equivalent \FO formula $\psi'$ by applying the following translation inductively. Since \WMPL restricts set quantification to finite paths, and every finite path starting from the root in a tree is uniquely identified by its finite endpoint $x$, we can represent the path $X$ as the set of all ancestors of $x$. Thus, we translate the quantifiers as follows: 
   \begin{itemize}
       \item first order quantification remains unchanged: $(\exists x. \varphi)' = \exists x. \varphi'$;
       \item second order quantification is replaced by first order quantifiers over the path's endpoint: $(\exists X. \varphi)' = \exists x. \varphi'$, where every atomic formula of the form $z \in X$ within $\varphi'$ is replaced by $z \leq x$.
   \end{itemize}
   By structural induction, it follows that $\psi$ and $\psi'$ are equivalent, demonstrating that \WMPL collapses into \FO over this class of structures.
\end{proof}


\begin{proposition}
    Decidability of \APT $\mapsto$ \AWT does not imply decidability of \MPL $\mapsto$ \FO.
\end{proposition}
\begin{proof}
     Take as input the \MPL formula for the language $T$ of Lemma \ref{MPL, WMSO but not FO}. Obtain an equivalent \APT, recognizing $T$ (this is always possible, since \APT is equivalent to \MSO and \MPL is just a fragment of \MSO). Recall that \AWT are equivalent in expressive power to \WMSO. Suppose to have an algorithm to decide \APT $\mapsto$ \AWT. By the proof of Lemma \ref{MPL, WMSO but not FO} and Rabin's theorem, the algorithm would output ``Yes'', but the language is not definable in \FO.
 \end{proof}

 \section{Proofs from Section \ref{PCTL = automaton}}

 \begin{proposition}
     \HWTcf are strictly more expressive than \FO.
 \end{proposition}

\begin{proof}
    In \cite[Theorem 12,14]{BDMMP25} it was shown that \FO is effectively equivalent to a subclass of \HWTcf. Together with Lemma \ref{T expressible in HWT}, this yields the Proposition.
\end{proof}

\pctlnorm*

\begin{proof}
    The proof is by structural induction on \PCTL formulae. The atomic and the boolean cases are trivial. Here we show the existential path quantifier cases (the universal ones are analogous):
    \begin{itemize}
        \item $\E\X \alpha$. It is easy to see that $\E\X \alpha$ is equivalent to the formula $\E\F(\Y\top \land \neg \Y\Y\top \land \alpha$).
 \item $\E\psi_1 \U \psi_2$. We focus on the path formula $\psi_1 \U \psi_2$.
    The intuitive idea is to see the latter formula as a \PLTL formula, by encoding
    the problematic state subformulae as fresh propositional letters, then rewrite it
    into an equivalent form $\F\eta$, and finally restore the original semantics of
    the problematic state subformulae (properly normalized according to the above grammar).

    We begin by removing the problematic state subformulae: we replace every formula
    $\varphi$ of the form $\E\alpha$ or $\A\alpha$ that appears inside $\psi_1$ or $\psi_2$ by a fresh
    atomic proposition $p_{\varphi}$.
    Call the formulae thus obtained $\psi_1'$ and $\psi_2'$,
    and note that these are pure past \PLTL formulae
    (i.e.~they do not use the operators $\X$ and $\U$). Moreover, notice that $\psi_1'\U\psi_2'$ is in negation normal form, since in \PCTL syntax we disallowed arbitrary negations.
    By \autoref{safeLTL with past is equal to safeLTL}, there is a
    \cosafeLTL formula equivalent to $\psi_1' \U \psi_2'$.
    Moreover, by Theorem \ref{safeLTL equivalent to safePLTL}, the latter
    \cosafeLTL formula is also equivalent to a formula $\F\eta$, where $\eta$ is a pure past \PLTL formula.

    Towards a conclusion, we replace inside $\F\eta$ each atomic proposition $p_{\varphi}$,
    introduced earlier, by a corresponding formula $\tilde \varphi$ that is equivalent to $\varphi$
    and that satisfies the desired grammar: this formula $\tilde \varphi$ can be obtained from $\varphi$
    by inductive hypothesis.
    In this way, we obtain a formula $\E\F\tilde\eta$ that is equivalent to $\E \psi_1 \U \psi_2$ and that
    satisfies the desired grammar.
        \item $\E\alpha \R \beta$. The proof is analogous to the case above, with the only difference that we use Proposition \ref{safeLTL equivalent to safePLTL} and Theorem \ref{safeLTL equivalent to safePLTL} to obtain a \safePLTL formula rather than a \cosafePLTL one.
     \end{itemize}
\end{proof}

\begin{lemma} [Lemma \ref{PCTL as CTLs fragment}, restated]
    The following logic, called \obbligationCTLs, is effectively equivalent to \PCTL:
    \begin{align*}
\varphi ~::= &~
  p
  ~\bigm|~
  \neg \varphi
  ~\bigm|~
  \varphi \lor \varphi
  ~\bigm|~
  \E \psi
  ~\bigm|~
  \E \chi
\\
\psi ~::= &~
  \varphi
  ~\bigm|~
  \psi \land \psi
  ~\bigm|~
  \psi \lor \psi
  ~\bigm|~
  \X \psi
  ~\bigm|~
  \psi \U \psi
  \\
  \chi ~::= &~
  \varphi
  ~\bigm|~
  \chi \land \chi
  ~\bigm|~
  \chi \lor \chi
  ~\bigm|~
  \X \chi
  ~\bigm|~
  \chi \R \chi
\end{align*}

\end{lemma}
 \begin{proof}
    The idea of the proof is the same of \cite[Theorem 3.5]{HT87}. While Hafer and Thomas' Theorem was a corollary of Gabbay's Separation Theorem, the above lemma is a corollary of the characterisations of \cosafeLTL and \safeLTL provided in \cite{chang1992characterization}. We prove that for every \PCTL formula in the normal form of Lemma \ref{PCTL normal form}, which is expressively complete w.r.t. \PCTL, there is an equivalent \obbligationCTLs formula. The other direction is virtually identical. We use the following grammar for \PCTL normal form:
    \begin{align*}
\varphi ~::= &~
  p
  ~\bigm|~
  \neg \varphi
  ~\bigm|~
  \varphi \lor \varphi
  ~\bigm|~
  \E\F \psi
  ~\bigm|~
  \E\G \psi
\\
\psi ~::= &~
  \varphi
  ~\bigm|~
  \psi \lor \psi
  ~\bigm|~
  \psi \land \psi
  ~\bigm|~
  \Y \psi
  ~\bigm|~
  \tilde{\Y} \psi
  ~\bigm|~
  \psi \s \psi
\end{align*}

    The proof is again by structural induction on state formulae. The atomic and boolean cases are trivial. We discuss the existential path quantifier cases:
    \begin{itemize}
        \item $\E\F\alpha$. Consider $\alpha$. We replace every formula
    $\varphi$ of the form $\E\F\beta$ or $\E\G\beta$ that appears inside $\alpha$ by a fresh
    atomic proposition $p_{\varphi}$. Call the formula thus obtained $\alpha'$. Since $\alpha'$ is a formula containing only past temporal operators, by Theorem \ref{safeLTL equivalent to safePLTL} there is a \cosafeLTL formula $\beta$ equivalent to $\F\alpha'$. The grammar for $\psi$ above is indeed \cosafeLTL, and so $\E\beta$ is a formula of \obbligationCTLs. Finally, we replace inside $\beta$ each atomic proposition $p_{\varphi}$,
    introduced earlier, by a corresponding formula $\tilde \varphi$ that is equivalent to $\varphi$
    and that satisfies the desired grammar: this formula $\tilde \varphi$ can be obtained from $\varphi$
    by inductive hypothesis. Hence, the resulting formula is still a \obbligationCTLs formula.
        \item $\E\G\alpha$. Apply the same reasoning as above, this time wrt \safeLTL. The resulting formula $\E\beta$, where $\beta$ is a \safeLTL formula, is a \obbligationCTLs formula, since the grammar for $\chi$ in the statement of the lemma is \cosafeLTL.
    \end{itemize}
\end{proof}

\begin{theorem}[One direction of Theorem \ref{thm:hwtcf&pctl}] \label{one direction}
    For every \HWTcf $\mathcal{A}$, there is an equivalent \PCTL formula $\varphi_\mathcal{A}$.
\end{theorem}

\begin{proof}
    Fix an \HWTcf $\mathcal{A} = \langle Q, \Sigma, \delta, q_I, F \rangle$. The proof will show that there is a translation $f$ that takes in input a formula $\theta \in \mathcal{B}^+(\{\diamond, \square\} \times Q)$ (or a state $q \in Q$) and outputs a \obbligationCTLs formula $f(\theta)$ (or $f(q)$) such that $\Lang(\mathcal{A}^\theta) = \Lang(f(\theta))$ (or $\Lang(\mathcal{A}^q) = \Lang(f(q))$). This implies that by taking in input $q_I$, $f(q_I)$ yields a formula equivalent to the automaton $\mathcal{A}$. In particular, we will translate the automaton in the negation normal form syntax of \obbligationCTLs, namely:
        \begin{align*}
\varphi ~::= &~
  p
  ~\bigm|~
  \neg p
  ~\bigm|~
  \varphi \lor \varphi
  ~\bigm|~
  \varphi \land \varphi
  ~\bigm|~
  \E \psi
  ~\bigm|~
  \E \chi
  ~\bigm|~
  \A \psi
  ~\bigm|~
  \A \chi
\\
\psi ~::= &~
  \varphi
  ~\bigm|~
  \psi \lor \psi
  ~\bigm|~
  \psi \land \psi
  ~\bigm|~
  \X \psi
  ~\bigm|~
  \psi \U \psi
  \\
  \chi ~::= &~
  \varphi
  ~\bigm|~
  \chi \lor \chi
  ~\bigm|~
  \chi \land \chi
  ~\bigm|~
  \X \chi
  ~\bigm|~
  \chi \R \chi
\end{align*}
We will assume that anytime we use the inductive hypothesis, the resulting formula is generated by the above grammar. Then, by Lemma \ref{PCTL as CTLs fragment}, it follows that there is a \PCTL formula equivalent to $\mathcal{A}$, yielding the theorem.

    We begin by showing how to translate a formula:
\begin{itemize}
    \item $f(\top) = \top$
    \item $f(\bot) = \bot$
    \item $f(\theta_1 \land \theta_2) = f(\theta_1) \land f(\theta_2)$
    \item $f(\theta_1 \lor \theta_2) = f(\theta_1) \lor f(\theta_2)$
    \item $f((\diamond,q))  = \E\X f(q)$
    \item $f((\square,q))  = \A\X f(q)$
\end{itemize}

 The rest of the proof will be devoted to show how to translate a state $q$. For every $\sigma \in \Sigma$, we define a \obbligationCTLs formula $\psi_\sigma$, such that
\begin{center}
    $ \psi_\sigma = \bigwedge_{p \in \sigma} p \land \bigwedge_{p \notin \sigma} \neg p$
\end{center}
    which states that the symbol $\sigma$ is being read.

    The proof is by induction on the order of the components $Q_i$ in the partitioning of Q. We distinguish the cases in which the state $q$ is in a transient, existential (accepting or rejecting) or universal (accepting or rejecting) component:
    \begin{itemize}
        \item If $q$ is in a transient component $Q_i$, then by construction, for every $\sigma \in \Sigma$, $\delta(q, \sigma) = \theta_{q, \sigma}$ is a formula in which appear only states from components of order lower than $Q_i$. Then, we can define the translation of $q$ as follows:
        \begin{center}
            $ f(q_i) = \bigvee_{\sigma \in \Sigma} ( \psi_\sigma \land f(\theta_{q, \sigma})) $
        \end{center}
    where the translation of every state in $\theta_{q, \sigma}$ is obtained by inductive hypothesis and the whole formula is translated according to the translation presented above. The soundness of the translation is easily seen.
    \item If $q$ is in an existential accepting component $Q_i$, we consider the linearisation $\mathcal{A}^{q}_{\mathrm{exit}}$ of $Q_i$ as defined in \cref{automaton and conjecture}. By construction it is an \NBW $\mathcal{A}^{q}_{\mathrm{exit}} = \langle Q_i \uplus \{q_{exit}\}, 2^{AP} \times 2^{B_{Q_i}}, \delta', q, Q_i \rangle$. We recall our assumptions on the existential accepting component $Q_i$ and its linearisation:
        \begin{enumerate}
            \item $\mathcal{A}^{q}_{\mathrm{exit}}$ is counter-free and looping;
            \item $2^{B_{Q_i}} \subseteq 2^{Atoms(\mathcal{A})}$, where $Atoms(\mathcal{A})$ is the set of all atoms possibly appearing in $\mathcal{A}$. Furthermore, every element $C \in 2^{B_{Q_i}}$ is a subset of atoms who can contain as states only states in components of order lower than $Q_i$;
            \item $Q_i$ is visible;
        \end{enumerate}

            \emph{Observation 1.} By (2), we can apply the inductive hypothesis on every atom appearing in an element $C \in 2^{B_{Q_i}}$. This means that, for every $C \in 2^{B_{Q_i}}$ we can obtain a \obbligationCTLs formula $f(\bigwedge C)$ such that $\Lang(\mathcal{A}^{\bigwedge C}) = \Lang(f(\bigwedge C))$, where $\bigwedge C$ is the conjunction of every atom appearing in $C$.

        \emph{Observation 2.} By the above observation and (3), it follows that for every $C, C' \in 2^{B_{Q_i}}$, we have $\Lang(f(\bigwedge C)) \cap \Lang(f(\bigwedge C')) = \emptyset$.

        Now, we reason on (1) to understand how to exploit it in order to obtain our translation. First, for every $C \in 2^{B_{Q_i}}$, define a fresh atomic proposition $p_C$. Let $AP_B$ be the union of $AP$ and of all these new atomic propositions. We then let $\mathcal{A}_B = \langle Q_i \cup \{q_{exit}\}, 2^{AP_B}, \delta_B, q, Q_i \rangle$ be $\mathcal{A}^{q}_{\mathrm{exit}}$ in which the alphabet is replaced by $2^{AP_B}$ and $\delta'$ is modified as we describe below. The intuition here is that we are ``hiding'' the atoms in $2^{B_{Q_i}}$ in the new alphabet, by the introduction of these fresh atomic propositions. Then, we define $\delta_B$ as follows: for every $q' \in Q_i$ and $\sigma' \in 2^{AP_B}$, if $\sigma' = \sigma \cup \{p_C\}$, for some $\sigma \in 2^{AP}$ and $C \in 2^{B_{Q_i}}$, then $\delta_B(q', \sigma') = \delta'(q', (\sigma, C))$, otherwise $\delta_B(q', \sigma') = \emptyset$.

        The automaton $\mathcal{A}_B$ thus obtained is obviously still counter-free, since the transition function mimics the transition function of the by assumption counter-free $\mathcal{A}^{q}_{\mathrm{exit}}$. Moreover, it is a looping \NBWcf over the alphabet $2^{AP_B}$.  By Lemma \ref{looping automata and results}, this implies that we can get a \safeLTL formula $\psi_B$ which is equivalent to $\mathcal{A}_B$. Notice that $\psi_B$ must be in negation normal form, since it is a \cosafeLTL formula.

        \textit{Observation 3.} For every $\Sigma$-tree $t$ and $q \in Q_i$, where $Q_i$ is an existential accepting component, $\mathcal{A}^q$ recognizes $t$ iff there is an infinite path $\pi$ in $t$ and an infinite word $w \in \Lang(\mathcal{A}^{q}_{\mathrm{exit}})$ of the form $w = (t(\pi(0)), C_0), (t(\pi(1)), C_1), \dots$, such that for every $j \geq 0$, $t_{\pi(j)}$ is recognized by $\mathcal{A}^{\bigwedge C_j}$.

        The above observation highlights the fact that if the \HWTcf starts in a state within an existential accepting component $Q_i$, then accepting the input tree reduces to the linearisation of $Q_i$ accepting the path ``selected'' by $Q_i$. We can formally reformulate the above observation as follows:

        \textit{Observation 4.} For every $\Sigma$-tree $t$ and $q \in Q_i$, where $Q_i$ is an existential accepting component, $\mathcal{A}^q$ recognizes $t$ iff there is an infinite path $\pi$ in $t$ and an infinite word $w \in \Lang(\mathcal{A}_B)$ over the alphabet $2^{AP_B}$, such that $w \models \psi_B$, and for all $i \geq 0$, (1) $w(i) \cap AP = t(\pi(i))$, (2) for all $p_C \in w(i)$, $t, (\pi(i)) \models f(\bigwedge C)$, and finally (3) there is a unique $C \in 2^{B_{Q_i}}$ such that $p_C \in w(i)$.

        Finally, we consider the formula $\psi_B$ obtained above, and replace every occurrence of the newly introduced atomic propositions $p_C$ by the \obbligationCTLs formula $f(\bigwedge C)$, that, as already argued above, can be obtained by inductive hypothesis. We denote the formula thus obtained by $f(\psi_B)$. Finally, we set the translation of our starting state $q$ to be:
        \begin{center}
            $f(q) = \E (f(\psi_B) \land \G(\bigvee_{C \in 2^{B_{Q_i}}} f(\bigwedge C)) ) $
        \end{center}
        Notice that $f(q)$ is in negation normal form.

        To conclude, we must prove that a $\Sigma$-tree $t \models f(q)$ iff $t$ is recognized by $\mathcal{A}^q$.

        Suppose $t \models f(q)$. Thus, there is an infinite path $\pi$ such that $t, \pi, 0 \models f(\psi_B)$ and there is a sequence $c = C_0, C_1, ...$ such that for all $j \geq 0, t, \pi(j) \models C_j$. We can then define an infinite word $w$ over $2^{AP_B}$ as follows for every $i \geq 0$: $w(i) = t(\pi(i)) \cup \{p_{C_i}\}$. In words, we label the $i$-th position in $w$ with the label of the $i$-th position in the path $\pi$ and a fresh atomic proposition associated with the element $C_i$. Now, we can consider the formula $\psi_B$ over this word, instead of $f(\psi_B)$. Thanks to Observation 4 above, we can conclude that, in order to prove that $\mathcal{A}^q$ recognizes $t$, we must simply prove that $w \models \psi_B$.

        The proof goes by structural induction on the formula $\psi_B$ (which is in negation normal form), by showing that for every subformula $\psi'_B$ of $\psi_B$ and every $i \geq 0$, if $t, \pi, i \models f(\psi'_B)$ then $w, i \models \psi'_B$. Every case trivializes by inductive hypothesis, except for the cases in which $\psi'_B$ is a literal $p_C$ or $\neg p_C$, associated with an element $C$.

        Suppose $\psi'_B = p_C$. Then, $f(\psi'_B) = f(\bigwedge C)$. Assume that $t, \pi, i \models f(\bigwedge C)$ and $t, (\pi(i)) \models f(\bigwedge C_i)$. By Observation 2, it must follow that $C = C_i$, showing that $p_C = p_{C_i}$. By construction of $w$, it follows that $w, i \models p_{C_i}$.

        Suppose $\psi'_B = \neg p_C$. Then, $f(\psi'_B) = \bigvee_{C' \in 2^{B_{Q_i}} \setminus \{C\}} f(\bigwedge C')$. Assume that $t, \pi, i \models \bigvee_{C' \in 2^{B_{Q_i}} \setminus \{C\}} f(\bigwedge C')$ and $t, (\pi(i)) \models f(\bigwedge C_i)$. By Observation 2, it must follow that $C \neq C_i$, showing that $p_C = \neg p_{C_i}$. By construction of $w$, it follows that $w, j \models \neg p_{C_i}$.

        The other direction of the proof, i.e., the one in which we prove that if $t$ is recognized by $\mathcal{A}^q$, then $t \models f(q)$ is similar thanks to Observation 4.

        \item If $q$ is in a universal rejecting component $Q_i$, we apply a dualization argument to reduce this case to the existential one. Let $\overline{\mathcal{A}} = \langle Q, \Sigma, \overline{\delta}, q_I, \overline{F}\rangle$ be the \HWTcf obtained by defining $\overline{F} = Q \setminus F$ and for every $q \in Q, \sigma \in \Sigma$, $\overline{\delta}(q, \sigma) = g(\delta(q, \sigma))$, where $g: (\{\diamond_k, \square_k\} \times Q) \rightarrow (\{\diamond_k, \square_k\} \times Q)$ is the following translation:

        \begin{itemize}
    \item $g(\top) = \bot$
    \item $g(\bot) = \top$
    \item $g(\theta_1 \land \theta_2) = g(\theta_1) \lor g(\theta_2)$
    \item $g(\theta_1 \lor \theta_2) = g(\theta_1) \land g(\theta_2)$
    \item $g((\diamond_k,q))  = (\square_k, q)$
    \item $g((\square_k,q))  = (\diamond_k, q)$
\end{itemize}
We call $\overline{\mathcal{A}}$ the \emph{dual} automaton of $\mathcal{A}$.
     $\mathcal{A}$ and $\overline{A}$ clearly share the same structural constraints, and every existential (resp., universal) component in $\mathcal{A}$ is a universal (resp., existential) one in $\overline{\mathcal{A}}$. Notice also that every accepting component in $\mathcal{A}$ is rejecting in $\overline{\mathcal{A}}$. Moreover, if $\mathcal{A}$ is counter-free and visible, $\overline{A}$ is, too, since no modification is performed at the structural level of the components. By an easy adapatation of \cite[Lemma 5.2]{kupferman2000automata}, it is easy to see that $\overline{\mathcal{A}}$ defines the complement language of $\Lang(\mathcal{A})$. Thus, we can consider $q$ in $\overline{A}$, that is a state in an \emph{existential accepting} component, and apply our translation $f(q)$ defined above. Since we have that $\Lang(\mathcal{A}^q) = \overline{\Lang (\overline{\mathcal{A}^q})}$, we also have that $\neg f(q)$ is equivalent to $\mathcal{A}^q$, and we are done.
     \item For the existential rejecting case and its dual we refer to
     \cite{BDMMP25}, but it is also quite easy to derive them by adapting the
     above proof.
    \end{itemize}
\end{proof}

\begin{theorem}[The other direction of Theorem \ref{thm:hwtcf&pctl}]
    For every \PCTL formula $\varphi$, there is an equivalent \HWTcf $\mathcal{A}_\varphi$.
\end{theorem}

\begin{proof}
    Given a \PCTL formula, by Lemma \ref{PCTL as CTLs fragment} we can obtain an equivalent \obbligationCTLs formula. In particular, we follow \cite[From Logics to Automata]{BBMP24} and employ a  \emph{simple form} of \obbligationCTLs, in which state formulae of the form $\E\psi$ in the grammar of \obbligationCTLs are replaced by state formulae of the form $\D^1\E\psi$. The semantics of a state formula $\D^1\E\psi$ is:
    \begin{itemize}
        \item $(t,u) \models_s \D^1 \E\psi$ iff $u$ has a child $v$ such that $(t, v) \models_s \E\psi$,
    \end{itemize}

    Hence the syntax we consider is:
        \begin{align*}
\varphi ~::= &~
  p
  ~\bigm|~
  \neg \varphi
  ~\bigm|~
  \varphi \lor \varphi
  ~\bigm|~
  \D^1\E \psi
  ~\bigm|~
  \D^1\E \chi
\\
\psi ~::= &~
  \varphi
  ~\bigm|~
  \psi \lor \psi
  ~\bigm|~
  \psi \land \psi
  ~\bigm|~
  \X \psi
  ~\bigm|~
  \psi \U \psi
  \\
  \chi ~::= &~
  \varphi
  ~\bigm|~
  \chi \lor \chi
  ~\bigm|~
  \chi \land \chi
  ~\bigm|~
  \X \chi
  ~\bigm|~
  \chi \R \chi
\end{align*}

    This simple form is expressively complete for \obbligationCTLs, as follows by a straightforward adaptation of \cite[Proposition B.4]{BBMP24}, and it is needed to enforce the visibility property (or, more precisely, to be able to use the inductive hypothesis). Hence, the proof goes by structural induction on \obbligationCTLs's simple form. We prove the claim as follows. Fix a \obbligationCTLs state formula $\varphi$ in simple form, and associate to each state subformula $\zeta$ an automaton $\mathcal{A}_\zeta$ as follows:

    \begin{itemize}
        \item $\zeta = p$. An automaton equivalent to an atomic proposition $p$ just needs one state. We can set $\mathcal{A}_\zeta = \langle \{q_p\}, 2^{AP}, \delta, q_p, \emptyset\rangle$, with $\delta (q_p, \sigma) = \top$, if $p \in \sigma$, and $\delta(q_p, \sigma) = \bot$, otherwise.
        \item $\zeta = \neg \zeta_1$. By inductive hypothesis, there is an automaton $\mathcal{A}_{\zeta_1}$ equivalent to $\zeta_1$. Then, it suffices to obtain the dual automaton $\overline{\mathcal{A}_{\zeta_1}}$, as defined in the proof of Theorem \ref{one direction}, which defines the complement language of $\Lang(\mathcal{A}_{\zeta_1})$.
        \item $\zeta = \zeta_1 \lor \zeta_2$. By inductive hypothesis, there are two automata $\mathcal{A}_{\zeta_1} = \langle Q_1, 2^{AP}, \delta_{\zeta_1}, q_{\zeta_1}, F_1 \rangle$ and $\mathcal{A}_{\zeta_2} = \langle Q_2, 2^{AP}, \delta_{\zeta_2}, q_{\zeta_2}, F_2 \rangle$, equivalent to $\zeta_1$ and $\zeta_2$, respectively. We define $\mathcal{A}_\zeta = \langle Q_1 \uplus Q_2 \uplus \{q_\zeta\}, 2^{AP}, \delta_\zeta, q_\zeta, F_1 \uplus F_2 \rangle$. The state set $Q_1 \uplus Q_2 \uplus \{q_\zeta\}$ can be partitioned in the components of the partition of $Q_1$, the components of the partition of $Q_2$ and a singleton transient component \{$q_\zeta$\}. The type of the components in $Q_1$ and $Q_2$ is unchanged from their type in the original automata. The ordering of the components is then arranged hierarchically: the singleton component $\{q_\zeta\}$ is the highest-order component, then come the components of $Q_1$, followed by those of $Q_2$. We define $\delta_\zeta$ as $\delta_{\zeta_1}$ (resp., $\delta_{\zeta_2}$), when applied to states in $Q_1$ (resp., $Q_2$), while, for every $\sigma \in 2^{AP}$, $\delta_\zeta(q_\zeta, \sigma) = \delta_\zeta(q_{\zeta_1}, \sigma) \lor \delta_\zeta(q_{\zeta_2}, \sigma)$. Finally, the set of accepting states is the union of the accepting states of $\mathcal{A}_{\zeta_1}$ and $\mathcal{A}_{\zeta_2}$. It is easy to see that this automaton recognizes the union of the languages of $\mathcal{A}_{\zeta_1}$ and $\mathcal{A}_{\zeta_2}$, while still remaining counter-free and visible.
       \item $\zeta = \D^1 \zeta_1$. By inductive hypothesis, there is an automaton $\mathcal{A}_{\zeta_1} = \langle Q, 2^{AP}, \delta, q_I, F \rangle$ equivalent to $\zeta_1$. We define $\mathcal{A}_\zeta = \langle Q \uplus \{q_\zeta\}, 2^{AP}, \delta_\zeta, q_\zeta, F \rangle$. As above, we partition the state set according to the partitioning of $Q$ and take \{$q_\zeta$\} as a singleton transient component. The components of $Q$ are left unchanged with respect to type and ordering, while \{$q_\zeta\}$ is taken as component of highest order. Notice that $F$ is unchanged from $\mathcal{A}_{\zeta_1}$.  Finally, $\delta_\zeta$ is defined as $\delta$ with respect to states in $Q$, while, for every $\sigma \in 2^{AP}$, we have $\delta_\zeta (q_\zeta, \sigma) = (\diamond, q_I)$.  The newly defined automaton scans the root and immediately sends to one of its children the starting state of the automaton $\mathcal{A}_{\zeta_1}$, that we have by inductive hypothesis. The soundness of this construction is clear.

        \item $\zeta = \E \gamma$, with $\gamma$ generated by the above $\chi$ grammar. While in the above cases we could use the inductive hypothesis, in the case of $\E\gamma$ we cannot, since $\gamma$ is not a \emph{state} subformula, but rather a \emph{path} subformula. We shall then proceed in a by now classical way, following \cite[Theorem 5.3]{kupferman2000automata}. First, let $max(\gamma)$ be the set of all state subformulae in $\gamma$ of the form $\D^1 \zeta_1$, which are not themselves subformulae of a formula of the form $\D^1 \zeta_2$. In a certain sense, they are the \emph{maximal} state subformulae of $\gamma$. Hence, we have $max(\gamma) = \{\D^1\zeta_1, \dots, \D^1\zeta_k\}$. By inductive hypothesis, for every $i \in [1, k]$, we have an automaton $\mathcal{A}_{\zeta_i} = \langle Q_i, 2^{AP}, \delta_i, q_i, F_i \rangle$ equivalent to $\zeta_i$. We also consider the dual versions $\overline{\mathcal{A}_{\zeta_i}} = \langle \overline{Q_i}, 2^{AP}, \overline{\delta_i}, \overline{q_i}, \overline{F_i} \rangle$ of each one of these automata. Specifically, we consider a \emph{renaming} of the dual automata, to enforce that all the state sets of these automata (the ones obtained by inductive hypothesis and their renamed duals) are disjoint.

        To proceed, associate to each formula $\D^1 \zeta_i$ a fresh atomic proposition $p_i$. Then, we call $AP_B$ the union of $AP$ and these newly introduced atomic propositions (the name is chosen only to have continuity with the proof of the above Theorem). Now, consider $\gamma$ and replace each formula $\D^1 \zeta_i$ with its associated fresh atomic proposition. Call the resulting formula $\gamma_B$. Clearly, $\gamma_B$ is a \safeLTL formula. By Proposition \ref{looping automata and results}, there is a looping \NBWcf $\mathcal{A}_{\gamma_B}$ equivalent to $\gamma_B$.

        \emph{Observation 1.} Given a $2^{AP}$-tree $t$, $t \models \E\gamma$ iff there is an infinite path $\pi$ in $t$ and an associated infinite word $w$ over $2^{AP_B}$ such that $w \in \Lang(\mathcal{A}_{\gamma_B})$, and the following holds for every $i \geq 0$, (1) $w(i) \cap AP = t(\pi(i))$, (2) for each $j \in [1, k]$, if $p_j \in w(i), $ then $t, \pi(i) \models \D^1 \zeta_j$ and (3) for each $j \in [1, k]$, if $p_j \notin w(i)$, then $t, \pi(i) \models \neg \D^1 \zeta_j$.

        Basically, the word $w$ captures entirely the information available along $\pi$, and of interest with respect to the formula $\E \gamma$. The automaton for $\E\gamma$ then will simulate the \NBWcf $\mathcal{A}_{\gamma_B}$ along a branch, while sending copies of the automata $\mathcal{A}_{\zeta_i}$ and $\overline{\mathcal{A}_{\zeta_i}}$ (for every $i \in [1, k]$) along this branch. We define formally $\mathcal{A}_\zeta = \langle Q \cup \bigcup\limits_{i=1}^{k} (Q_i \cup \overline{Q_i}), 2^{AP}, \delta_\zeta, q_I, F \rangle$.
        $Q$ forms an existential component, which can be seen as the state set of the \NBWcf $\mathcal{A}_{\gamma_B}$. The other states are just the state set of the automata obtained by inductive hypothesis and their renamed duals. The partitioning of this state set is performed according to the partitioning of the original automata, also with respect to types of the components, while the new states in $Q$ form, as already said, a unique existential component. For what concerns the ordering, $Q$ is the component of highest order, while all the others can preserve their original ordering. Then, we have that $\delta_\zeta$ simulates, for every $i \in [1, k]$, $\delta_i$ (resp., $\overline{\delta_i}$) when in a state of $Q_i$ (resp., $\overline{Q_i}$). For every state $q \in Q$ and $\sigma \in 2^{AP}$, on the other hand, $\delta_\zeta$ is defined as follows:
        \[
        \delta_\zeta(q, \sigma) = \bigvee_{B \subseteq [1, k]} \bigvee_{q' \in \delta_{\gamma_B}(q, B(\sigma))} ((\diamond, q') \land \bigwedge_{i \in B} (\diamond, q_i) \land \bigwedge_{i \in [1,k] \setminus B} (\square, \overline{q_i}))        \]
        where $B(\sigma) = \sigma \cup \bigcup_{i \in B} \{p_i\}$ and $\delta_{\gamma_B}$ is just the transition function of the \NBWcf $\mathcal{A}_{\gamma_B}$. In words, the above transition function states that there is a transition of the \NBWcf $\mathcal{A}_{\gamma_B}$ from the state $q$ to the state $q'$ performed by reading the symbol in $2^{AP_B}$ that is the union of $\sigma$ and the fresh atomic propositions $p_j$, for every $j \in B$. The following conjunction then just exploits the automata obtained by inductive hypothesis and their duals to transform the requirement given over the alphabet $2^{AP_B}$ to a requirement over the alphabet $2^{AP}$. Finally, $F$ is just the union of all the accepting states in the automata $\mathcal{A}_{\zeta_i}$ and $\overline{\mathcal{A}_{\zeta_i}}$, for every $i \in [1, k]$.

        By the inductive hypothesis and Observation 1, we obtain that $\mathcal{A}_\zeta$ is visible and $\Lang(\mathcal{A}_\zeta) = \Lang(\E\gamma)$. To obtain counter-freeness, we just need to show that the linearisation of $Q$ over the alphabet $2^{AP} \times 2^{B_{Q}}$ is counter-free and looping. The looping part is trivial, since we have that no state in $Q$ is in $F$. For the counter-freeness this can be done by showing a straightforward bijection between the alphabet $2^{AP} \times 2^{B_Q}$ of the linearisation of $Q$ and the alphabet $2^{AP_B}$ of the \NBWcf $\mathcal{A}_{\gamma_B}$. In particular, for every $\sigma' \in 2^{AP_B}$ of the form $\sigma' = \sigma \cup \bigcup_{i \in I} p_i$, for some $I \subseteq [1,k]$, we can associate a corresponding symbol in $2^{AP} \times 2^{B_Q}$ that is $(\sigma, \bigcup_{i \in I} (\diamond, q_i) \cup \bigcup_{i \in [1,k] \setminus I} (\square, \overline{q_i}))$. Clearly, we can go also in the other direction. Then, it is easy to see that one can exploit this bijection to show that the transition function of the linearisation of $Q$ behaves as $\delta_{\gamma_B}$ over the state set of $\mathcal{A}_{\gamma_B}$ (i.e., $Q$ again), and derive from the counter-freeness of $\mathcal{A}_{\gamma_B}$ the counter-freeness of the linearisation of $Q$.
        \item $\zeta = \E \gamma$, for $\gamma$ generated by the above $\psi$ grammar. This is done as above, replacing \safeLTL by \cosafeLTL and looping \NBWcf by \NCWcf.
        \end{itemize}
\end{proof}

\begin{corollary} \label{Universal apctl}
    For every universal (resp., existential) \HWTcf, there is an equivalent \APCTL (resp., \EPCTL) formula.
\end{corollary}

\begin{proof}
    It suffices to inspect the proof of Theorem \ref{one direction}, and notice that when we translate components in which only disjunctions and atoms of the form $(\diamond, q)$ appear in the transition function, we always end up in the fragment that only uses existential path quantification. The same holds for universal \HWTcf.
\end{proof}

\mainthm*

\begin{proof}
    Given an \MPL formula $\varphi$, one can obtain an equivalent \APT. Suppose to have an algorithm to decide \APT $\mapsto$ \AWT. Then, by running the algorithm on the \APT equivalent to $\varphi$, one can decide if the formula is \AWT recognizable. If it is, by the equivalence of \AWT and \WMSO, $\varphi$ is also definable in \WMSO. Thus, $\Lang(\varphi)$ is a language in $\MPL \cap \WMSO$. Then, by Conjecture \ref{con:comfrg}, there is an \HWTcf recognizing $\Lang(\varphi)$, and by Theorem \ref{thm:hwtcf&pctl}, also a \PCTL formula defining it. If $\varphi$ is instead not \WMSO definable, for a similar reasoning it follows that $\Lang(\varphi)$ is not \PCTL definable. The same holds if the input is given as a \CTLs formula, because of the already mentioned \cite[Theorem 3.6]{HT87}.
\end{proof}

\section{Proofs from Section \ref{Common fragment}}

\Wordtotree*

\begin{proof}
    Given $\mathcal{A} = \langle Q, \Sigma, \delta, q_I, F \rangle$ as in the statement of the Lemma, we show how to construct an \HWTcf recognizing $\triangle [\Lang(\mathcal{A})]$.

    First, we construct a UWW $\mathcal{B}$ equivalent to $\mathcal{A}$. We construct $\mathcal{B}$ by taking two copies of $\mathcal{A}$. In particular, define $\mathcal{B} = \langle (Q\times\{1\}) \cup (Q \setminus F \times \{2\}), \Sigma, \delta', (q_I, 1), (Q\times \{1\}) \rangle$. Hence, we take a first copy of $Q$ in its entirety, and a second copy of $Q$ devoid of final states. Moreover, every state in the first copy is deemed accepting, while all the others are not. Finally, we define $\delta'$ as follows:
    \begin{itemize}
        \item if $\delta(q, \sigma) = q' \notin F$, $\delta'((q, 1), \sigma$) = $\{(q',1), (q', 2)\}$ and $\delta'((q,2), \sigma) = \{(q', 2)\}$,
        \item if $\delta(q, \sigma) = q' \in F, \delta' ((q,1), \sigma) = \{(q', 1)\}$ and $\delta'((q,2), \sigma) = \emptyset$.
    \end{itemize}
    
    The intuition behind these constraints is simple: whenever the original automaton goes to a non final state, $\mathcal{B}$ spawns a run in the second copy of $Q$ (or, if it's starting from a state in the second copy, it keeps staying in the second copy); on the other hand, whenever the original automaton goes into a final state, $\mathcal{B}$ either keeps staying in the first copy (which is entirely accepting), or deems an entire run as accepting, if it was in the second copy.

    $\mathcal{B}$ is clearly weak, since ($Q\times \{1\}$) and $(Q\setminus F \times \{2\})$ are two different components, the former entirely accepting, the latter entirely rejecting, and there is no transition from $(Q \setminus F \times \{2\})$ to $(Q \times \{1\})$. Moreover, $\mathcal{B}$ is counter-free. This is because:
    \begin{itemize}
        \item the restriction of $\mathcal{B}$ to the first copy is isomorphic to $\mathcal{A}$ as a transition structure, and so it is counter-free by assumption;
        \item the transitions from the first copy to the second one cannot create counters, because a counter requires a closed path of states. Since there is no transition from the second copy back to the first, any closed path must be entirely contained within either the first copy or the second copy;
        \item the second copy is an induced subgraph of $Q$: since $\mathcal{A}$ is deterministic, a counter is strictly defined by the presence of a non-trivial cyclic path of deterministic transitions. Removing states  removes transitions and paths; it cannot create new paths or new cycles. Therefore, any counter in the second copy would necessarily be a counter in the original automaton $\mathcal{A}$. Since $\mathcal{A}$ is counter-free, the second copy must also be counter-free.
    \end{itemize}

    Finally, we have $\Lang(\mathcal{A}) = \Lang(\mathcal{B})$.
    We show both inclusions:
    \begin{itemize}
        \item $\Lang(\mathcal{A}) \subseteq \Lang(\mathcal{B})$. Let $w \in \Lang(\mathcal{A})$. This means that the run of $\mathcal{A}$ over $w$ visits infinitely often states in $F$. Consider all the possible runs of $\mathcal{B}$ over $w$. First, there is a run staying forever in the first copy, simply mirroring all the transitions of $\mathcal{A}$. This run is clearly accepting, because the first copy is entirely accepting. Then, there is a  possibly infinite number of other runs, which are the runs going into the second copy every time $\mathcal{A}$ goes into a non final state. Nevertheless, all these other runs are also accepting. By assumption, the run of $\mathcal{A}$ on $w$ visits infinitely often states in $F$. Since a state in the second copy of $Q$ in $\mathcal{B}$ deems a run accepting whenever the corresponding state in $\mathcal{A}$ would have gone to a final state, this implies, together with the assumption that $\mathcal{A}$ accepts $w$, that all the runs of $\mathcal{B}$ on $w$ going in the second copy will eventually be deemed accepting.
        \item $\Lang(\mathcal{B}) \subseteq \Lang(\mathcal{A})$. Suppose $w \in \Lang(\mathcal{B})$. This means that every run of $\mathcal{B}$ over $w$ is deemed accepting. By mirroring the reasoning performed above the inclusion follows easily.
    \end{itemize}

    To conclude the proof, we consider the tree automaton $\mathcal{\triangle B}$, which is obtained from $\mathcal{B}$ by transforming every state $(q, n)$ in the output of the transition function of $\mathcal{B}$, for $n \in \{0,1\}$, in an atom $(\square, (q, n))$. Clearly, $\triangle[\Lang(\mathcal{B})]  = \Lang(\mathcal{\triangle B})$. Moreover $\mathcal{\triangle B}$ is :
    \begin{itemize}
        \item weak, because $\mathcal{B}$ is, too,
        \item hesitant, because the first and the second copy satisfy the constraints of the universal components,
        \item counter-free, because the first and the second components are both counter-free, as argued above,
        \item visible, since this constraint trivializes in the absence of disjunction.
    \end{itemize}

    This concludes the proof.
\end{proof}



\msocmn*

\begin{proof}
    Let $\varphi$ be an \MSO sentence. Then, by the equivalence with \APT, it is possible to obtain an \APT automaton recognizing $\mathcal{L}(\varphi)$. It is then easy to check if $\mathcal{L}(\mathcal{A})$ is a language of the form $\triangle[L]$, for some regular $\omega$-word language $L$ (see, e.g., \cite[Theorem 8]{bojanczyk2008common}). If it is not, then $\mathcal{L}(\varphi)$ is not in \LTL $\cap$ \PCTL. Otherwise, it suffices to check if $L$ is recognizable by a \DBWcf. As already argued, this is decidable. Then, if $L$ is recognizable by a \DBWcf, $\mathcal{L}(\varphi)$ is in \LTL $\cap$ \PCTL, otherwise, it is not.
\end{proof}

\section{Proofs from Section \ref{no current}}

\pctlwmtl*

\begin{proof}
    Consider \PCTL syntax as follows:
    \begin{align*}
\varphi ~::= &~
  p
  ~\bigm|~
  \neg \varphi
  ~\bigm|~
  \varphi \lor \varphi
  ~\bigm|~
  \E\X \varphi
  ~\bigm|~
  \E \varphi \U \varphi
  ~\bigm|~
  \A \varphi \U \varphi
  ~\bigm|~
  \Y \varphi
  ~\bigm|~
  \varphi \s \varphi
\end{align*}

This syntax is equivalent to the one provided in the main paper. We inductively define a translation $f_x$, where $x$ is a first order variable, that takes in input a \PCTL formula and outputs an equivalent \WMTL formula with one free first order variable $x$ as follows:
\begin{itemize}
    \item $f_x(p) = p(x)$, for every $p \in AP$,
    \item $f_x(\neg\varphi) = \neg f_x(\varphi)$,
    \item $f_x(\varphi_1 \lor \varphi_2) = f_x(\varphi_1) \lor f_x(\varphi_2)$,
    \item $f_x(\E\X\varphi) = \exists y (child(x,y) \land f_y(\varphi))$,
    \item $f_x(\E(\varphi_1 \U \varphi_2) =  \exists y ( x \le y \land f_y(\varphi_2) \land \forall z ( (x \le z < y) \rightarrow f_z(\varphi_1) ))$
    \item $f_x(\A\varphi_1\U\varphi_2) = \exists X ( x \in X \land \forall y \in X (f_y(\varphi_2)  \lor (f_y(\varphi_1) \land \forall z (child(y, z) \rightarrow z \in X))))$
    \item $f_x (\Y\varphi) = \exists y (child(y,x) \land f_y(\varphi))$
    \item $f_x(\varphi_1 \s \varphi_2) = \exists y ( y \le x \land f_y(\varphi_2) \land \forall z ( (y \le z < x) \rightarrow f_z(\varphi_1) ))$
\end{itemize}

where $child (y, x) = x > y \land \neg \exists w (x > w > y)$.
The soundness of the translation is easy to check. This shows that \PCTL can be translated soundly in \WMTL, and thus it is at most as expressive as the latter.
\end{proof}

\wmtlmpl*

\begin{proof}
    \MPL $\not\subseteq$ \WMTL follows from the fact that \MPL $\not\subseteq$ \WMSO. Consider for example the language ``on every path, $p$ occurs finitely often'', over the alphabet $2^{\{p,q\}}$. By \cite[Lemma 7]{rabin1970weakly}, this language is not \WMSO definable, but it is clearly \MPL definable as follows:
    \begin{center}
        $\neg \exists X \forall x (x \in X \rightarrow \exists y(x < y \land p(y) \land y \in X))$
    \end{center}

    Hence, it remains to show that \WMTL $\not\subseteq$ \MPL.     Consider again the alphabet $2^{\{p,q\}}$. The separating language will be $U =$ ``there is a finite subtree containing the root such that: its leaves are labeled $\{p,q\}$, its internal nodes labeled $\{q\}$ have at least a child and its internal nodes not labeled $\{q\}$ have exactly two children''. This language is definable in \WMTL as follows (notice that in the following formula the ``exactly two children'' requirement is fulfilled under the assumption that we are considering full binary trees, i.e., trees in which every node has exactly two children):
    \[
        \exists X. root \in X \land \forall z \in X \rightarrow \land
        \begin{cases}
            (\neg p(z) \rightarrow \exists y(child(y,z) \land y \in X) \\
            (\neg q(z) \rightarrow \forall y(child(y,z) \rightarrow y \in X)
        \end{cases}
    \]

    Our claim is that this language is not \MPL definable. We will show it exploiting two families of trees. We first introduce these two families of trees using a rather straightforward definition. Then, we provide a more technically involved definition that is specifically tailored for the subsequent proof.
    Hence, let us define a \emph{$\Sigma$-multicontext} as a finite or infinite $\Sigma \cup \{*_1, \dots, *_k\}$-tree, for $k \in \mathbb{N}$, in which every internal node is labeled by a symbol of $\Sigma$ and every leaf (if any) is labeled by $*_i$, for some $i \in [1,k]$. Consider the following $2^{\{p,q\}}$-multicontext $c$:

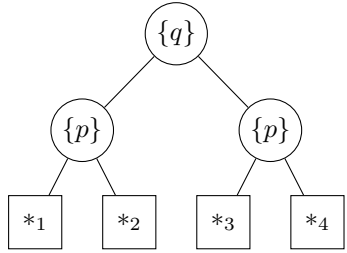
  \begin{center}
\begin{forest}
for tree={
draw,
circle,         
minimum size=2em,
inner sep=2pt,
l=3em,          
s sep=1.5em     
},
where n children=0{shape=rectangle}{}
[$\{q\}$
[$\{p\}$
[$\ast_1$]
[$\ast_2$]
]
[$\{p\}$
[$\ast_3$]
[$\ast_4$]
]
]
\end{forest}
\end{center}
    Then, let $l$ be the $2^{\{p,q\}}$-tree whose root is labeled by $\{q\}$ and all its other nodes are labeled by $\emptyset$, and $r$ be the $2^{\{p,q\}}$-tree whose root is labeled by $\{p,q\}$ and all its other nodes are labeled by $\emptyset$. In pictures:
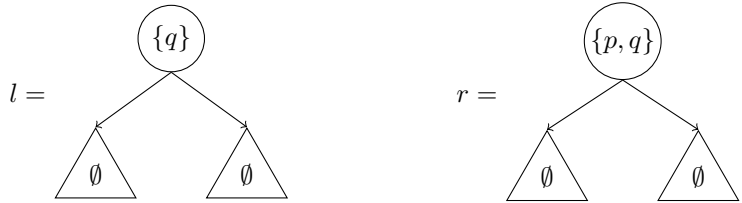
\begin{center}

$l =$
\begin{tikzpicture}[
    baseline=(current bounding box.center),
    level distance=1.8cm,
    sibling distance=2cm,
    edge from parent/.style={draw, ->},
    edge from parent path={(\tikzparentnode.south) -- (\tikzchildnode.north)}
  ]
  \node[draw, circle, minimum size=2.5em, inner sep=1pt] {$\{q\}$}
    child { node[draw, regular polygon, regular polygon sides=3, minimum size=3.5em, inner sep=0pt] {$\emptyset$} }
    child { node[draw, regular polygon, regular polygon sides=3, minimum size=3.5em, inner sep=0pt] {$\emptyset$} };
\end{tikzpicture}
\hspace{2cm} 
$r =$
\begin{tikzpicture}[
    baseline=(current bounding box.center),
    level distance=1.8cm,
    sibling distance=2cm,
    edge from parent/.style={draw, ->},
    edge from parent path={(\tikzparentnode.south) -- (\tikzchildnode.north)}
  ]
  \node[draw, circle, minimum size=2.5em, inner sep=1pt] {$\{p,q\}$}
    child { node[draw, regular polygon, regular polygon sides=3, minimum size=3.5em, inner sep=0pt] {$\emptyset$} }
    child { node[draw, regular polygon, regular polygon sides=3, minimum size=3.5em, inner sep=0pt] {$\emptyset$} };
\end{tikzpicture}

\end{center}
 We then define $t_1$ to be $c$ in which the leaf labeled $*_1$ is substituted by $l$ and the leaves labeled $*_2, *_3$ and $*_4$ are substituted by $r$, while $t'_1$ is defined as $c$ in which the leaves labeled $*_1$, $*_3$ are substituted by $l$ while the leaves labeled $*_2$ and $*_4$ are substituted by $r$. In pictures:
 \begin{center}
$t_1 =$
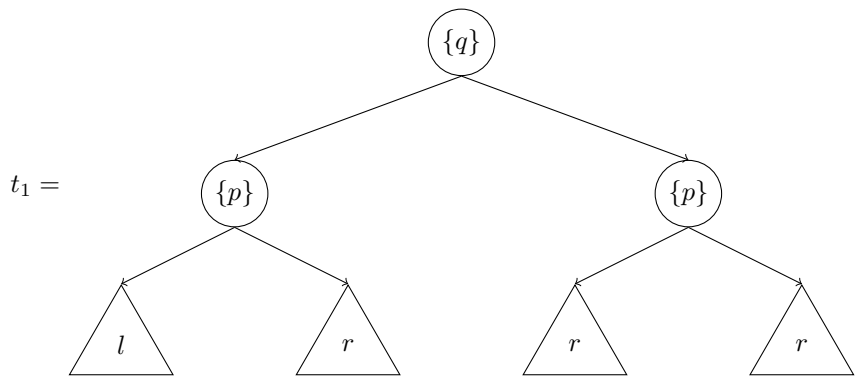
\begin{tikzpicture}[
    baseline=(current bounding box.center),
    level 1/.style={sibling distance=6cm, level distance=2cm},
    level 2/.style={sibling distance=3cm, level distance=2cm},
    edge from parent/.style={draw, ->},
    edge from parent path={(\tikzparentnode.south) -- (\tikzchildnode.north)}
  ]
  \node[draw, circle, minimum size=2.5em, inner sep=1pt] {$\{q\}$}
    child { node[draw, circle, minimum size=2.5em, inner sep=1pt] {$\{p\}$}
        child { node[draw, regular polygon, regular polygon sides=3, minimum size=4.5em, inner sep=0pt] {$l$} }
        child { node[draw, regular polygon, regular polygon sides=3, minimum size=4.5em, inner sep=0pt] {$r$} }
    }
    child { node[draw, circle, minimum size=2.5em, inner sep=1pt] {$\{p\}$}
        child { node[draw, regular polygon, regular polygon sides=3, minimum size=4.5em, inner sep=0pt] {$r$} }
        child { node[draw, regular polygon, regular polygon sides=3, minimum size=4.5em, inner sep=0pt] {$r$} }
    };
\end{tikzpicture}

\vspace{2cm} 

$t'_1 =$
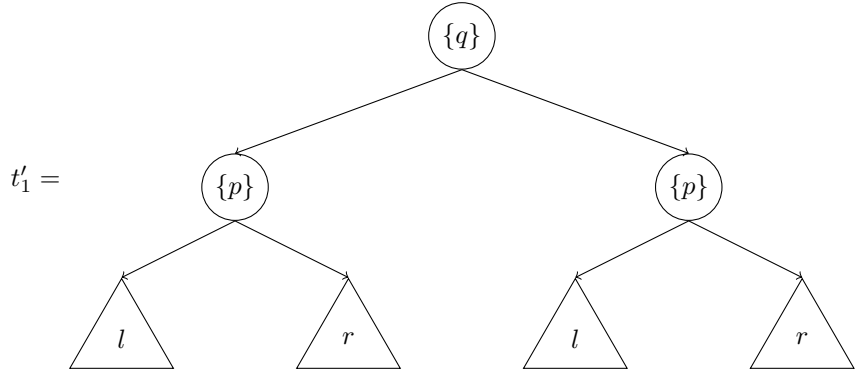
\begin{tikzpicture}[
    baseline=(current bounding box.center),
    level 1/.style={sibling distance=6cm, level distance=2cm},
    level 2/.style={sibling distance=3cm, level distance=2cm},
    edge from parent/.style={draw, ->},
    edge from parent path={(\tikzparentnode.south) -- (\tikzchildnode.north)}
  ]
  \node[draw, circle, minimum size=2.5em, inner sep=1pt] {$\{q\}$}
    child { node[draw, circle, minimum size=2.5em, inner sep=1pt] {$\{p\}$}
        child { node[draw, regular polygon, regular polygon sides=3, minimum size=4.5em, inner sep=0pt] {$l$} }
        child { node[draw, regular polygon, regular polygon sides=3, minimum size=4.5em, inner sep=0pt] {$r$} }
    }
    child { node[draw, circle, minimum size=2.5em, inner sep=1pt] {$\{p\}$}
        child { node[draw, regular polygon, regular polygon sides=3, minimum size=4.5em, inner sep=0pt] {$l$} }
        child { node[draw, regular polygon, regular polygon sides=3, minimum size=4.5em, inner sep=0pt] {$r$} }
    };
\end{tikzpicture}
\end{center}

 It is easy to see that $t_1 \in U$ and $t'_1 \notin U$. We then define infinite such trees as follows: for every $i \ge 2$, $t_i$ is defined to be $c$ in which the leaf labeled $*_1$ is substituted by $t'_{i-1}$ and the leaves labeled $*_2,*_3$ and $*_4$ are substituted by $t_{i-1}$, while $t'_i$ is defined to be $c$ in which the leaves labeled $*_1$ and $*_3$ are substituted by $t'_{i-1}$, while the leaves labeled $*_2$ and $*_4$ are substituted by $t_{i-1}$. Again, for every $i \ge 2$, $t_i \in U$ and $t'_i \notin U$. We will call these two families of trees $T$ and $T'$, respectively. The underlying intuition is that for every path $\pi$ in $t_i$ (and symmetrically for $t'_i$), there exists a corresponding path $\pi'$ in $t'_i$ that not only shares the exact same sequence of labels, but also enters an \emph{isomorphic subtree} starting from (at most) its third position. For instance, suppose a CTL$^*$ formula $\E\alpha$ is satisfied by $t_i$; this implies there is a path $\pi$ such that $t_i, \pi, 0 \models \alpha$. Moreover, suppose that $\pi$ is a path defined by its first two turns (e.g., an initial left turn followed by a right turn). Because the subtrees reached in $t_i$ and $t'_i$ after these first two steps are isomorphic, there must exist a matching path $\pi'$ in $t'_i$ that perfectly mirrors the structure and labeling of $\pi$ from that point onward, guaranteeing that $t'_i, \pi', 0 \models \alpha$ holds as well. For the first two positions of such a path, we will instead rely on the identity of the labels and inductive hypothesis (once an induction is defined). This intuition will be formalized in a claim we will prove after the introduction of the \CTLs syntax we are going to use and a new definition of the families of trees $T$ and $T'$, better suited for the subsequent proof. 

 We consider $\CTLs$, generated by the following syntax:
 \begin{align*}
\varphi ~::= &~
  p
  ~\bigm|~
  \neg \varphi
  ~\bigm|~
  \varphi \lor \varphi
  ~\bigm|~
  \E \psi
\\
\psi ~::= &~
  \varphi
  ~\bigm|~
  \neg \psi
  ~\bigm|~
  \psi \lor \psi
  ~\bigm|~
  \X \psi
  ~\bigm|~
  \psi \U \psi
\end{align*}

 This is equivalent to the syntax provided in the main paper. Given a \CTLs (state or path) formula $\varphi$, we denote by $|\varphi|$ the \emph{size} of $\varphi$, that is considered as the length of the string for representing $\varphi$. 
 To present the claim to prove that \CTLs cannot define $U$, we define two further families of trees $S$ and $S'$, and redefine $T$ and $T'$ in terms of $S$ and $S'$. To do so, we use the following $2^{\{p,q\}}$-multicontexts $c'$ and $d$:

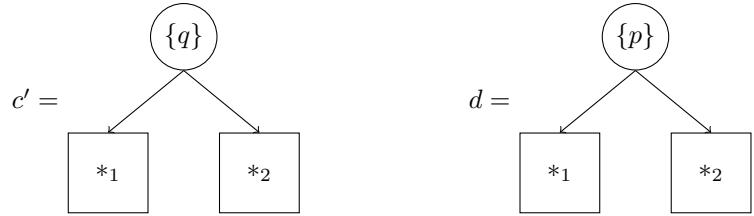
\begin{center}

$c' =$
\begin{tikzpicture}[
    baseline=(current bounding box.center),
    level distance=1.8cm,
    sibling distance=2cm,
    edge from parent/.style={draw, ->},
    edge from parent path={(\tikzparentnode.south) -- (\tikzchildnode.north)}
  ]
  \node[draw, circle, minimum size=2.5em, inner sep=1pt] {$\{q\}$}
    child { node[draw, rectangle, minimum size=3em, inner sep=2pt] {$\ast_1$} }
    child { node[draw, rectangle, minimum size=3em, inner sep=2pt] {$\ast_2$} };
\end{tikzpicture}
\hspace{2cm} 
$d =$
\begin{tikzpicture}[
    baseline=(current bounding box.center),
    level distance=1.8cm,
    sibling distance=2cm,
    edge from parent/.style={draw, ->},
    edge from parent path={(\tikzparentnode.south) -- (\tikzchildnode.north)}
  ]
  \node[draw, circle, minimum size=2.5em, inner sep=1pt] {$\{p\}$}
    child { node[draw, rectangle, minimum size=3em, inner sep=2pt] {$\ast_1$} }
    child { node[draw, rectangle, minimum size=3em, inner sep=2pt] {$\ast_2$} };
\end{tikzpicture}
\end{center}
Notice that $c'$ is simply a truncated version of $c$. Now, we define $s_1$ and $s'_1$ to be $d$ in which we substitute the leaf labeled by $*_1$ with $l$ and $*_2$ with $r$, while $s'_1$ is defined as $d$ in which we substitute both $*_1$ and $*_2$ with $r$. In pictures:
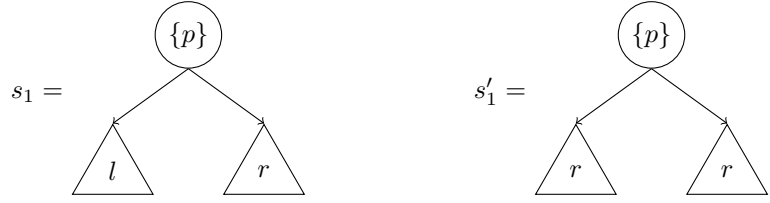
\begin{center}

$s_1 =$
\begin{tikzpicture}[
    baseline=(current bounding box.center),
    level distance=1.8cm,
    sibling distance=2cm,
    edge from parent/.style={draw, ->},
    edge from parent path={(\tikzparentnode.south) -- (\tikzchildnode.north)}
  ]
  \node[draw, circle, minimum size=2.5em, inner sep=1pt] {$\{p\}$}
    child { node[draw, regular polygon, regular polygon sides=3, minimum size=3.5em, inner sep=0pt] {$l$} }
    child { node[draw, regular polygon, regular polygon sides=3, minimum size=3.5em, inner sep=0pt] {$r$} };
\end{tikzpicture}
\hspace{2cm} 
$s'_1 =$
\begin{tikzpicture}[
    baseline=(current bounding box.center),
    level distance=1.8cm,
    sibling distance=2cm,
    edge from parent/.style={draw, ->},
    edge from parent path={(\tikzparentnode.south) -- (\tikzchildnode.north)}
  ]
  \node[draw, circle, minimum size=2.5em, inner sep=1pt] {$\{p\}$}
    child { node[draw, regular polygon, regular polygon sides=3, minimum size=3.5em, inner sep=0pt] {$r$} }
    child { node[draw, regular polygon, regular polygon sides=3, minimum size=3.5em, inner sep=0pt] {$r$} };
\end{tikzpicture}
\end{center}

It follows that we can redefine $t_1$ and $t'_1$ to be:
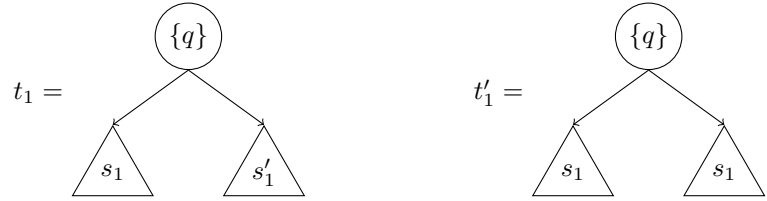
\begin{center}

$t_1 =$
\begin{tikzpicture}[
    baseline=(current bounding box.center),
    level distance=1.8cm,
    sibling distance=2cm,
    edge from parent/.style={draw, ->},
    edge from parent path={(\tikzparentnode.south) -- (\tikzchildnode.north)}
  ]
  \node[draw, circle, minimum size=2.5em, inner sep=1pt] {$\{q\}$}
    child { node[draw, regular polygon, regular polygon sides=3, minimum size=3.5em, inner sep=0pt] {$s_1$} }
    child { node[draw, regular polygon, regular polygon sides=3, minimum size=3.5em, inner sep=0pt] {$s'_1$} };
\end{tikzpicture}
\hspace{2cm} 
$t'_1 =$
\begin{tikzpicture}[
    baseline=(current bounding box.center),
    level distance=1.8cm,
    sibling distance=2cm,
    edge from parent/.style={draw, ->},
    edge from parent path={(\tikzparentnode.south) -- (\tikzchildnode.north)}
  ]
  \node[draw, circle, minimum size=2.5em, inner sep=1pt] {$\{q\}$}
    child { node[draw, regular polygon, regular polygon sides=3, minimum size=3.5em, inner sep=0pt] {$s_1$} }
    child { node[draw, regular polygon, regular polygon sides=3, minimum size=3.5em, inner sep=0pt] {$s_1$} };
\end{tikzpicture}
\end{center}
i.e., as $c'$, in which the leaf labeled $*_1$ is substituted by $s_1$ and the leaf labeled $*_2$ by $s'_1$, while $t'_1$ is $c'$ in which both leaves are substituted by $s_1$.  

It is easy to predict the following: for every $i \ge 2$, we define $s_i$ to be $d$ in which the leaf labeled $*_1$ is substituted by $t_{i-1}$ and the one labeled $*_2$ by $t'_{i-1}$, while $s'_i$ is defined to be $d$ in which both leaves are substituted by $t_{i-1}$. Hence, for every $i$, $t_i$ can be redefined as $c'$ in which the leaf labeled $*_1$ is substituted by $s'_1$, while the other leaf is substituted by $*_2$, while $t'_i$ is redefined to be $c'$ in which both leaves are substituted by $s_i$.   Again in pictures, for $i \ge 2$, $s_i$ and $s'_i$ are defined as follows:

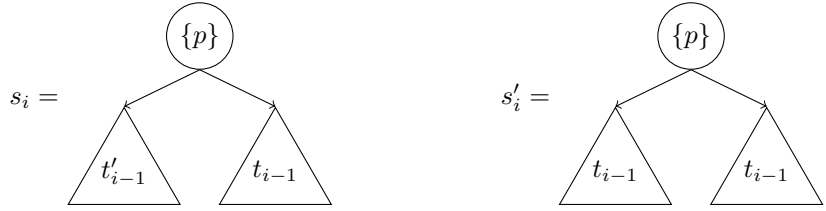
\begin{center}

$s_i =$
\begin{tikzpicture}[
    baseline=(current bounding box.center),
    level distance=1.8cm,
    sibling distance=2cm,
    edge from parent/.style={draw, ->},
    edge from parent path={(\tikzparentnode.south) -- (\tikzchildnode.north)}
  ]
  \node[draw, circle, minimum size=2.5em, inner sep=1pt] {$\{p\}$}
    child { node[draw, regular polygon, regular polygon sides=3, minimum size=3.5em, inner sep=0pt] {$t'_{i-1}$} }
    child { node[draw, regular polygon, regular polygon sides=3, minimum size=3.5em, inner sep=0pt] {$t_{i-1}$} };
\end{tikzpicture}
\hspace{2cm} 
$s'_i =$
\begin{tikzpicture}[
    baseline=(current bounding box.center),
    level distance=1.8cm,
    sibling distance=2cm,
    edge from parent/.style={draw, ->},
    edge from parent path={(\tikzparentnode.south) -- (\tikzchildnode.north)}
  ]
  \node[draw, circle, minimum size=2.5em, inner sep=1pt] {$\{p\}$}
    child { node[draw, regular polygon, regular polygon sides=3, minimum size=3.5em, inner sep=0pt] {$t_{i-1}$} }
    child { node[draw, regular polygon, regular polygon sides=3, minimum size=3.5em, inner sep=0pt] {$t_{i-1}$} };
\end{tikzpicture}
\end{center}

while $t_i$ and $t'_i$ are as follows:

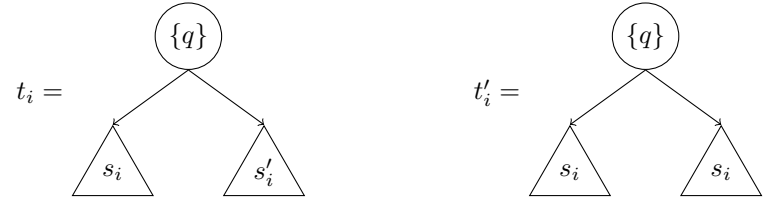
\begin{center}

$t_i =$
\begin{tikzpicture}[
    baseline=(current bounding box.center),
    level distance=1.8cm,
    sibling distance=2cm,
    edge from parent/.style={draw, ->},
    edge from parent path={(\tikzparentnode.south) -- (\tikzchildnode.north)}
  ]
  \node[draw, circle, minimum size=2.5em, inner sep=1pt] {$\{q\}$}
    child { node[draw, regular polygon, regular polygon sides=3, minimum size=3.5em, inner sep=0pt] {$s_i$} }
    child { node[draw, regular polygon, regular polygon sides=3, minimum size=3.5em, inner sep=0pt] {$s'_i$} };
\end{tikzpicture}
\hspace{2cm} 
$t'_i =$
\begin{tikzpicture}[
    baseline=(current bounding box.center),
    level distance=1.8cm,
    sibling distance=2cm,
    edge from parent/.style={draw, ->},
    edge from parent path={(\tikzparentnode.south) -- (\tikzchildnode.north)}
  ]
  \node[draw, circle, minimum size=2.5em, inner sep=1pt] {$\{q\}$}
    child { node[draw, regular polygon, regular polygon sides=3, minimum size=3.5em, inner sep=0pt] {$s_i$} }
    child { node[draw, regular polygon, regular polygon sides=3, minimum size=3.5em, inner sep=0pt] {$s_i$} };
\end{tikzpicture}
\end{center}

Following standard binary tree conventions, we use $0$ to denote a left turn and $1$ to denote a right turn. We still consider unordered trees, we just use this convention to make the presentation clearer. Thus, a path $\pi$ of the form $01\pi'$ is a path that starts at the root, descends to the left child, moves to its right child, and then continues along the remaining path $\pi'$. Let $\Pi_{t_i}$, $\Pi_{t'_i}$, $\Pi_{s_i}$, and $\Pi_{s'_i}$ denote the sets of paths in trees $t_i$, $t'_i$, $s_i$, and $s'_i$, respectively. 

To proceed, we define four families of mapping functions $f_i, f'_i, g_i,$ and $g'_i$, for all $i \ge 1$, which maps paths between the unprimed and primed trees back and forth. The crucial intuition here is that while a function takes a path in one tree and outputs a path in a completely different tree, the input and output sequences remain absolutely indistinguishable with respect to their labels. To emphasize this identical labeling, we employ a slight but deliberate abuse of notation: we reuse the symbol $\pi$ on both sides of the mapping. Here, $\pi$ acts as the tail of the input path in the source tree, and that exact same sequence of labels $\pi$ constitutes the tail of the output path in the target tree. Formally, for any initial step $\sigma \in \{0,1\}$ and remaining path sequence $\pi$, we define:
\begin{itemize}
    \item $g_i: \Pi_{s_i} \rightarrow \Pi_{s'_i}$ defined as $g_i(0\pi) = \sigma f'_{i-1}(\pi)$ and $g_i(1\pi) = \sigma\pi$;
    \item $g'_i: \Pi_{s'_i} \rightarrow \Pi_{s_i}$ defined as $g'_i(\sigma \pi) = 1 \pi$;
    \item $f_i : \Pi_{t_i} \rightarrow \Pi_{t'_i}$ defined as $f_i(0\pi) = \sigma\pi$ and $f_i(1\pi) = \sigma g'_{i-1}(\pi)$;
    \item $f'_i: \Pi_{t'_i} \rightarrow \Pi_{t_i}$ defined as $f'_i(\sigma\pi) = 0\pi$.
\end{itemize}

It is easy to see that the output path of these functions is identically labeled to the input path. Moreover, let us observe also the following, for every $m \ge 1$: 
\begin{itemize}
    \item for every $i \ge 1$, for every $\pi \in \Pi_{s_m}$, the subtree of $s_m$ rooted at $\pi(i)$ is isomorphic to the subtree of $s'_m$ rooted at $g_m(\pi)(i)$;
    \item for every $i \ge 1$, for every $\pi \in \Pi_{t'_m}$, the subtree of $t'_m$ rooted at $\pi(i)$ is isomorphic to the subtree of $t_m$ rooted at $f'_m(\pi)(i)$;
    \item for every $i \ge 2$, for every $\pi \in \Pi_{s'_m}$, the subtree of $s'_m$ rooted at $\pi(i)$ is isomorphic to the subtree of $s'_m$ rooted at $g'_m(\pi)(i)$;
    \item for every $i \ge 2$, for every $\pi \in \Pi_{t_m}$, the subtree of $t_m$ rooted at $\pi(i)$ is isomorphic to the subtree of $t'_m$ rooted at $g_m(\pi)(i)$.
\end{itemize}

Again, the above follows by a straightforward inspection of the families of trees we considered. 

Finally, we can state and prove the claim that yields the theorem. Notice that a path formula can also be a state one, according to our grammar. Nevertheless, we separate the two cases for clarity.

\begin{claim}
    Let $\alpha$ be a \CTLs formula. Then, for all $m \ge |\alpha|$:
    \begin{enumerate}
        \item if $\alpha$ is a state formula, then $t_m \models \alpha$ iff $t'_m \models \alpha$ and $s_m \models \alpha$ iff $s'_m \models \alpha$;
        \item if $\alpha$ is a path formula, then:
        \begin{itemize}
            \item for all $\pi \in \Pi_{t_m}$, $t_m, \pi, 0 \models \alpha$ iff $t'_m, f_m(\pi), 0 \models \alpha$;
            \item for all $\pi \in \Pi_{t'_m}$, $t'_m, \pi, 0 \models \alpha$ iff $t_m, f'_m(\pi), 0 \models \alpha$;
            \item for all $\pi \in \Pi_{s_m}$, $s_m, \pi, 0 \models \alpha$ iff $s'_m, g_m(\pi), 0 \models \alpha$;
            \item for all $\pi \in \Pi_{s'_m}$, $s'_m, \pi, 0 \models \alpha$ iff $s_m, g'_m(\pi), 0 \models \alpha$.
        \end{itemize}
    \end{enumerate}
\end{claim}

\begin{claimproof}
    By induction on $|\alpha|$. 
    \begin{itemize}
        \item Base: $|\alpha| = 1$. Hence $|\alpha|$ is an atomic proposition in both items' case.
        \begin{enumerate}
            \item Clearly, for every $m > 1$, the roots of $t_m$ and $t'_m$, as well as those of $s_m$ and $s'_m$, are identically labeled, thereby being indistinguishable by an atomic proposition. 
            \item As above. 
        \end{enumerate}
        \item Step: $|\alpha| > 1$. The inductive hypothesis states that the claim holds for all formulas $\beta$ such that $|\beta| < |\alpha|$. 
        \begin{enumerate}
            \item If $\alpha$ is a state formula, boolean cases are trivially handled by inductive hypothesis. The only interesting case is the existential path quantifier:
            \begin{itemize}
                \item $\alpha = \E\beta$. We show only the case assuming $t_m \models \E\beta$, for some $m \ge |\alpha|$, all the others being analogous. Suppose $t_m \models \E\beta$. Then, for some path $\pi \in \Pi_{t_m}$, $t_m, \pi, 0 \models \beta$. Since $|\beta| < |\alpha|$, we can apply the inductive hypothesis, and by item 2 of the claim it follows that $t_m, \pi, 0 \models \beta$ iff $t'_m, f_m(\pi), 0 \models \beta$. Hence $t_m, \pi(0) \models \E\beta$ iff $t'_m, f_m(\pi)(0) \models \E\beta$, implying $t_m \models \E\beta$ iff $t'_m \models \E\beta$, since $\pi(0)$ and $f_m(\pi)(0)$ denote the roots of $t_m$ and $t'_m$, respectively. 
            \end{itemize}
            \item If $\alpha$ is a path formula, again boolean cases are trivial by inductive hypothesis. Hence, we handle temporal cases:
            \begin{itemize}
                \item $\alpha = \X \beta$. Suppose $t_m, \pi, 0 \models \X\beta$. By $\X$ semantics, this implies $t_m, \pi, 1 \models \beta$. We distinguish two cases: one in which $\pi$ is of the form $0\pi'$, and one in which $\pi$ is of the form $1\pi'$:
                \begin{itemize}
                    \item if $\pi$ is of the form $0\pi'$, then $f_m(\pi)$ is just the identity function. Moreover, the subtree rooted at $\pi(1)$ is isomorphic to the subtree rooted at $f_m(\pi)(1)$. It follows that $t_m, \pi, 1 \models \beta$ iff $t'_m, f_m(\pi), 1 \models \beta$ (if $\beta$ is a path formula as well as a state one), implying that $t_m, \pi, 0 \models \X\beta$ iff $t'_m, f_m(\pi), 0 \models \X\beta$. 
                    \item if $\pi$ is of the form $1\pi'$, then $f_m(\pi)$ = $\sigma g'_m(\pi')$. Clearly, $t_m, \pi, 1 \models \beta$ is equivalent to $s_m, \pi', 0 \models \beta$, for such a path. By inductive hypothesis, $s_m, \pi', 0 \models \beta$ iff $s'_m, g'_m(\pi'), 0 \models \beta$. This is equivalent to $t_m, \pi, 1 \models \beta$ iff $t'_m, f_m(\pi), 1 \models \beta$, by construction of the trees, and we can conclude $t_m, \pi, 0 \models \X\beta$ iff $t'_m, f_m(\pi), 0 \models \X\beta$. 
                \end{itemize}
                The other cases are either easier or analogous, and always follow by the properties of our construction and of the mappings. 
                \item $\alpha = \beta_1 \U \beta_2$. Suppose $t_m, \pi, 0 \models \beta_1 \U \beta_2$. By $\U$ semantics, this is equivalent to $t_m, \pi, 0 \models \beta_2 \lor (\beta_1 \land \X(\beta_1\U\beta_2)$. Again we distinguish two cases:
                \begin{itemize}
                    \item if $\pi$ is of the form $0\pi'$, then $f_m(\pi)$ is just the identity function. If $t_m, \pi, 0 \models \beta_2$ then by inductive hypothesis $t_m, \pi, 0 \models \beta_2$ iff $t'_m, f_m(\pi), 0 \models \beta_2$, implying $t_m, \pi, 0 \models \beta_1 \U \beta_2$ iff $t'_m, f_m(\pi), 0 \models \beta_1 \U \beta_2$. On the other hand, if $t_m, \pi, 0 \models \beta_1 \land \X(\beta_1\U\beta_2)$, then $t_m, \pi, 0 \models \beta_1$ and $t_m, \pi, 0 \models \X(\beta_1\U\beta_2)$. By inductive hypothesis, $t_m, \pi, 0 \models \beta_1$ iff $t'_m, f_m(\pi), 0 \models \beta_1$.                Moreover, the subtree rooted at $\pi(1)$ is isomorphic to the subtree rooted at $f_m(\pi)(1)$. It follows that $t_m, \pi, 1 \models \beta_1 \U \beta_2$ iff $t'_m, f_m(\pi), 1 \models \beta_1 \U \beta_2$, implying that $t_m, \pi, 0 \models \X(\beta_1\U\beta_2)$ iff $t'_m, f_m(\pi), 0 \models \X(\beta_1\U\beta_2)$, yielding the claim. 
                    \item if $\pi$ is of the form $1\pi'$, then $f_m(\pi) = \sigma g'_m(\pi')$. Unfolding further the semantics of $\U$, notice that $t_m, \pi, 0 \models \beta_1 \U \beta_2$  is equivalent to $t_m, \pi, 0 \models \beta_2 \lor (\beta_1 \land \X\beta_2) \lor (\beta_1 \land \X\beta_1 \land \X\X(\beta_1 \U\beta_2))$. Then:
                    \begin{itemize}
                        \item if $t_m, \pi, 0 \models \beta_2$, the case is handled easily by inductive hypothesis;
                        \item if $t_m, \pi, 0 \models \beta_1 \land \X\beta_2$, the $\beta_1$ part is handled by inductive hypothesis and the $\X\beta_2$ by using the already proved case of $\X$;
                        \item if $t_m, \pi, 0 \models \beta_1 \land \X\beta_1 \land \X\X(\beta_1 \U \beta_2))$, the $\beta_1$ and $\X\beta_1$ parts are proved as immediately above, while to prove that $t_m, \pi, 0 \models \X\X(\beta_1 \U \beta_2)$, we simply notice that the subtree rooted at $\pi(2)$ and $f_m(\pi)(2)$ are isomorphic, hence $t_m, \pi, 2 \models \beta_1 \U \beta_2$ iff $t'_m, f_m(\pi), 2 \models \beta_1 \U \beta_2$, and we are done. 
                    \end{itemize}
                \end{itemize}
                All the remaining cases are proved by similar reasoning.
            \end{itemize}
        \end{enumerate}
        \end{itemize} 
\end{claimproof}
This claim shows that \CTLs, and analogously \MPL, by their equivalence, cannot define a \WMTL-definable language, namely $U$.
\end{proof}

\end{document}